\documentclass[acmsmall]{acmart}

\setcopyright{none}

\AtBeginDocument{%
  }

\usepackage{algorithm}
\usepackage{algorithmic}
\usepackage{amsthm}
\usepackage{mathtools}
\allowdisplaybreaks[4]
\usepackage{subcaption}

\newtheorem{theorem}{Theorem}

\newtheorem{definition}[theorem]{Definition}

\theoremstyle{remark}
\newtheorem{rem}{Remark}

\author{Fatemeh Fardno}
\email{ffardno2@illinois.edu}
\affiliation{%
  \institution{University of Illinois Urbana-Champaign}
  \city{Champaign}
  \country{USA}
}
\author{S. Rasoul Etesami}
\email{etesami1@illinois.edu}
\affiliation{%
  \institution{University of Illinois Urbana-Champaign}
  \city{Champaign}
  \country{USA}
}

\begin{document}

\title{Dynamic Transaction Scheduling and Pricing in the
Ethereum Mempool}

\begin{abstract}
    The Ethereum blockchain utilizes the EIP-1559 algorithm to manage transaction inclusion and block assembly. However, the EIP-1559 algorithm and much of the existing literature address this problem from a static perspective, emphasizing how prices evolve over time without considering the dynamic evolution of transactions within the mempool. Motivated by this framework, we study a dynamic transaction scheduling problem in which transactions have heterogeneous sizes and per-unit values and arrive in the mempool over time. While prior work typically assumes that transactions are impatient, we instead model them as patient, reflecting the fact that unscheduled transactions remain in the mempool and may be scheduled in future blocks.

    To explicitly account for the stochastic dynamics of transactions waiting in the mempool, we model the problem as a discounted Markov Decision Process (MDP) that captures both arrival dynamics and the evolving state of the transaction pool. To this end, we first provide a novel primal–dual interpretation of the static EIP-1559 algorithm, which allows us to view the pricing scheme as the dual variables of an appropriate linear program that maximizes the total value of scheduled transactions. This perspective enables us to characterize the interaction among pricing, block capacity constraints, and transaction selection. Building on this interpretation, we extend the framework by presenting a dual-based formulation of the transaction scheduling problem in the dynamic setting, interpreting block prices as decision variables linked to the dual occupancy measures of the underlying MDP. Moreover, we formulate the objective as maximizing the long-run discounted reward while accounting for holding costs and capacity overshoot penalties. We employ the Natural Policy Gradient (NPG) algorithm to determine the optimal scheduling policy. Our results show that, as the overshoot penalty increases, the average scheduled transaction volume converges to the target block capacity, and the resulting NPG policy updates closely resemble the EIP-1559 price update rule \cite{babaioff2024optimality} in the static setting.

    We also study two special cases of the MDP formulation: one with homogeneous transactions and one with uniform arrivals. In the homogeneous case, where the protocol directly controls the scheduled transaction volume, we show that the optimal policy has a threshold structure. Motivated by this characterization, we then propose a bang--bang pricing mechanism for uniform arrivals to address the challenge of block overflow and derive a theoretical lower bound on the block capacity required to ensure system stability. Finally, we validate our theoretical findings through extensive numerical experiments.
\end{abstract}




\maketitle
\section{Introduction} 
Blockchains are distributed, tamper-resistant digital ledgers that record transactions across a network of computers in a way that ensures transparency, security, and decentralization \cite{nakamoto2008bitcoin}. Blockchains, such as Ethereum \cite{buterin2014ethereum}, operate by repeatedly scheduling transactions into blocks. For simplicity, one may assume that each transaction has a size and a value per size for its user, and that each block has a capacity limit. Blocks are assembled by validators, whom the protocol designer seeks to incentivize to maximize the total value of included transactions subject to this capacity constraint. A key challenge is that both users and validators are selfish and strategic, so the transaction selection mechanism must account for such behavior. Ethereum addresses this challenge through the \emph{EIP-1559} mechanism \cite{eip1559}, which relaxes the strict maximum block size from $B$ to $c \cdot B$ for some $c > 1$, and instead effectively constrains the average block size to the target value $B$. As a result, EIP-1559 produces a sequence of blocks whose sizes may vary over time, but whose average is regulated to remain close to $B$.

The EIP-1559 mechanism operates as follows. For each block $t$, the mechanism sets a price $p_t$ and accepts only transactions whose value per unit size exceeds $p_t$. Transactions that do not meet this threshold remain in the mempool until they are scheduled at a later time. We assume that transactions are \emph{patient}, meaning they can be scheduled at any future time without losing value. If the total size of eligible transactions exceeds the maximum block size $c \cdot B$, then transactions are prioritized according to the tips they are willing to pay. In this paper, we do not consider tips in our analysis and instead assume that, whenever the capacity constraint $c \cdot B$ is violated, an adversary selects which transactions are included in the block until the capacity $c \cdot B$ is reached.

The price is updated according to a simple rule \cite{babaioff2024optimality}: if the total size of scheduled transactions exceeds the target block capacity $B$, the price for the next block is increased; otherwise, it is decreased. More formally, the price evolves according to
\[
p_{t+1} = \max\left\{ p_{\min}, \, p_t \cdot e^{\eta \frac{Q_t - B}{B}} \right\},
\]
where $Q_t$ denotes the total size of transactions scheduled in block $t$, and $\eta$ is a small constant, set in Ethereum to $\eta = \tfrac{1}{8}$. The protocol designer also specifies a minimum price $p_{\min}$ to ensure that the block price never falls below this level. This pricing mechanism, while simple and effective as shown in \cite{babaioff2024optimality}, does not incorporate information about transactions waiting in the mempool or about newly arriving transactions. Such information can be important in settings where the cost of holding transactions in the mempool is high, in which case it may be desirable to schedule more than the target block size in order to prevent the mempool from becoming congested.

We first formulate the transaction scheduling problem as a linear program with the objective of maximizing the total value of scheduled transactions. We then show that the block price arises naturally as a dual variable of this linear program. This observation allows us to extend EIP-1559 to a dynamic setting with patient users using a Markov Decision Process (MDP), in which the action taken by the mechanism at each time $t$ is the price set for block $t$.

More specifically, to explicitly capture the stochastic dynamics of transactions waiting in the mempool, we model the protocol as the decision-making agent in this MDP. The state corresponds to the mempool configuration, that is, the number of transactions of each value and size currently present, and the action is the price chosen by the protocol. The reward is defined to include penalties for mempool congestion as well as for scheduling transactions in excess of the target block capacity. Our objective is to maximize the long-run discounted average reward, balancing the benefit of reducing the mempool congestion while still aiming to meet the target block capacity. We then apply the Natural Policy Gradient (NPG) algorithm to compute an optimal pricing policy. Our results show that, as the penalty for exceeding the target block capacity increases, the average volume of scheduled transactions converges to that of the EIP-1559 algorithm, while the resulting policy updates closely resemble the EIP-1559 price update rule.

We finally study two special cases of the MDP formulation. First, we consider a homogeneous setting in which all transaction sizes and values are identical. In this case, pricing becomes irrelevant because all transactions have the same value, and the problem reduces to one in which the protocol directly controls the scheduled transaction volume. In this case, we show that the optimal policy has a threshold structure. In particular, as long as mempool congestion remains below a certain threshold, the protocol schedules the minimum possible volume. Once congestion reaches the threshold, the protocol schedules more than $B$ units to reduce congestion, continuing until congestion falls below the threshold, at which point the cycle repeats.

Subsequently, we study another special case of the MDP in which arrivals are uniform, meaning that at each time step exactly one transaction of each type arrives in the system. Motivated by the simple threshold structure of the optimal scheduling policy in the homogeneous transaction setting, we propose and analyze a bang--bang pricing mechanism for the uniform arrival setting; that is, a mechanism in which the price is restricted to either its minimum or maximum value. In particular, we derive a theoretical lower bound on the block capacity $B$ required to ensure system stability under the bang--bang pricing mechanism with uniform arrivals.

\subsection{Related Works}

In the Ethereum blockchain, the transaction fee mechanism determines the price that users pay to have their transactions included in the blockchain. In recent years, there has been significant progress in the design and analysis of blockchain transaction fee mechanisms, particularly within Ethereum. A foundational line of work was initiated by Vitalik Buterin \cite{buterin2014ethereum, ButerinBlockchainRP}, who proposed a set of ideas aimed at controlling block size through protocol-controlled pricing. Central among these ideas is EIP-1559, a transaction fee mechanism originally proposed by Buterin \cite{eip1559} and later implemented in Ethereum. The EIP-1559 mechanism has since been analyzed in academic works, including the paper by Tim Roughgarden \cite{roughgarden2020transaction}, which provides the first game-theoretic analysis of EIP-1559, examining its incentive properties and the economic rationale behind the price update rule.

More recently, Babaioff et al. \cite{babaioff2024optimality} studied the optimality properties of the EIP-1559 mechanism. While their analysis provides strong theoretical support for EIP-1559 in static settings, it does not explicitly model the dynamic evolution of transactions in the mempool over time. Several other works study different aspects of the EIP-1559 mechanism. For example, Reijsbergen et al. \cite{reijsbergen2021transaction} show that, while EIP-1559 achieves its goal on average, its short-term dynamics can be chaotic. In a different direction, Hougaard et al. \cite{hougaard2023farsighted} analyze EIP-1559 in settings where miners behave strategically rather than passively following the protocol.

Another related work, \cite{angeris2024multidimensional}, also formulates blockchain resource allocation as an optimization problem and interprets prices through dual variables. However, their analysis focuses on regret with respect to the single best fixed price in hindsight, whereas we provide a competitive ratio bound against the optimal total value of scheduled transactions, which is a stronger and more challenging benchmark to achieve. Our work differs more fundamentally in that, by interpreting prices as dual variables, we extend the framework to a dynamic setting with full mempool information. In this setting, we show that, in the absence of congestion penalties, the resulting update rule coincides with EIP-1559.

This body of work primarily adopts a static perspective, focusing on how prices evolve in response to previous block utilization and strategic user behavior, rather than on the stochastic evolution of transactions in the mempool. In contrast to this literature, our work explicitly models the dynamics of the mempool, treating unscheduled transactions as state variables that directly influence long-term scheduling decisions.

There is also a large body of literature studying dynamic pricing in EIP-1559 \cite{leonardos2023optimality,ferreira2021dynamic}. Crapis et al. \cite{crapis2024optimal} formulate blockchain fee design as a sequential decision-making problem and derive optimal fee update rules, which are empirically compared to EIP-1559 using Ethereum data. In this framework, the protocol that sets transaction fees acts as the learning agent; however, the model does not satisfy the Markov property. Interestingly, they show that the resulting optimal policy resembles the EIP-1559 update rule, with some key differences. In particular, EIP-1559 uses the current block size as a naive estimate for future demand, whereas their optimal policy relies on an explicit predictive estimate.

Leonardos et al. \cite{leonardos2021dynamical} model transaction arrivals to the mempool as a Poisson process with impatient transactions, meaning that they leave the pool if they are not included in the next block. They analyze the price update rule as a discrete-time, discrete-space stochastic process $\{p_t\}_{t \geq 0}$, where the source of randomness is the number of transactions included in block $t$ given the price $p_t$. They study the convergence and stability of the resulting dynamics and derive bounds on the step size of the price update rule that guarantee global convergence to equilibrium.

From a different perspective, \cite{math13061010} studies the EIP-1559 mechanism using tools from queuing theory. Similar to our approach, the mempool is modeled as a collection of queues, with the system state defined by the lengths of these queues. However, rather than maintaining separate queues for each transaction type as in our approach, transactions are classified into high-priority and low-priority classes, each with its own queue. Within this framework, the authors derive key performance metrics, including stability conditions, the stationary distribution, and the average queue length and waiting time for each transaction class.

Similar to our approach, \cite{madrigal2025pricing} analyzes the EIP-1559 mechanism using a Markovian framework. Under suitable assumptions, the authors show that the price process $\{p_t\}_{t \ge 0}$ satisfies the Markov property and can be modeled as a discrete-time Markov chain. This characterization allows them to study the long-run behavior of the price.

\subsection{Contributions}
The key contributions of this work are as follows:
\begin{itemize}
    \item We present a novel primal--dual interpretation of the EIP-1559 algorithm, which enables us to view its pricing scheme as the dual variables of an appropriate linear program that maximizes the total value of scheduled transactions. Using this principled primal--dual framework, we analyze the competitive ratio of the EIP-1559 algorithm.
    \item Building upon the primal--dual analysis, we extend the framework to a dynamic setting by modeling the state evolution as a Markov decision process (MDP).
    \item We employ an episodic Natural Policy Gradient (NPG) method to determine the optimal scheduling policy in the MDP. In particular, we show that the exponential price update rule used in EIP-1559 emerges naturally as a special case of the NPG algorithm, thereby providing a formal justification for the mechanism’s design. We numerically demonstrate that as the penalty for exceeding the target block capacity increases, the average scheduled transaction volume converges to that of the EIP-1559 algorithm.
    \item We analytically study two special cases of the MDP formulation: one with homogeneous transactions and one with uniform arrivals. In the homogeneous transaction setting, we show that the optimal policy has a threshold structure. Motivated by this result, we propose a simple bang--bang pricing mechanism for the uniform arrival setting and use it to derive a lower bound on the target block capacity $B$ required to ensure stability under such policies.
    \item We further validate our theoretical findings through extensive numerical experiments.
\end{itemize}

\section{Static Online Scheduling Problem}

In this section, we present an online transaction scheduling problem in the static setting, i.e., without any state evolution, and extend this model to a dynamic setting in Section~\ref{sec:dynamic}.

Consider an online transaction scheduling problem in which time is divided into periods $t = 1,\dots,T$, each corresponding to a block with fixed capacity $B$. At each time $t$, a set of transactions arrives. Each transaction $j$ is characterized by a size $q_j$ and a per-unit value $v_j$, resulting in a total value of $q_j v_j$ if it is scheduled in a block. Transaction $j$ arrives at time $r_j$ and is \emph{patient}, meaning that it can be scheduled in any block $t \ge r_j$ without loss of value. The objective is to schedule transactions across blocks so as to maximize the total value of scheduled transactions over all blocks.

\begin{rem}
In the EIP-1559 mechanism, the strict maximum block size is relaxed to $2B$, while the price update rule effectively constrains the average block size to the target value $B$. A difficulty arises when there are too many transactions with high per-unit values. Under EIP-1559, the total size of any block is constrained to be at most twice the target block size. In this case, we assume that an adversary selects a maximal (by inclusion) set of transactions that fit within the allowed maximum block size, i.e., $2B$. 
\end{rem}

In the offline setting, when all transaction arrival times and values/sizes are given \emph{a priori}, we can write a linear program relaxation that maximizes the total value of scheduled transactions subject to block capacity constraints and feasibility constraints on transaction assignments. To this end, let $x_{jt} \in [0,1]$ denote the fraction of transaction $j$ scheduled in block $t$. The resulting (primal) linear program relaxation is as follows
\begin{align}
\label{eq:primal}
    \max \quad & \sum_{j \in \mathcal{C}} \sum_{t=1}^T q_j v_j x_{jt} \\
    \text{s.t.} \quad 
    & \sum_{t=r_j}^T x_{jt} \le 1 \quad \forall j \notag\\
    & \sum_{j \in \mathcal{C}} x_{jt} q_j \le 2B \quad \forall t \notag\\
    & x_{jt} \ge 0 \quad \forall j \in \mathcal{C},\ t \ge r_j, \notag
\end{align}
where $\mathcal{C}$ is the set of all transactions. Subsequently, the dual program can be written as
\begin{align}
\label{eq:dual}
    \min \quad & \sum_{t=1}^T \beta_t + \sum_{j \in \mathcal{C}} \alpha_j \\
    \text{s.t.} \quad & q_j v_j \le \frac{\beta_t q_j}{2B} + \alpha_j 
    \quad \forall j\in \mathcal{C},\ \forall t \ge r_j \notag\\
    & \beta_t \ge 0,\quad \alpha_j \ge 0 \quad \forall j,t. \notag
\end{align}

We now provide an interpretation of the dual variables. Let $\beta_t$ denote the price of using the full capacity $2B$ of block $t$. Scheduling $q_j$ units of transaction $j$ on block $t$ therefore incurs a cost of $\frac{\beta_t}{2B}\, q_j$. Transaction $j$ has total value $v_j q_j$. Thus, if it is scheduled on block $t$, its utility is $q_j v_j - \frac{\beta_t}{2B} q_j$. Define
\[
\alpha_j \coloneqq \max_{t \ge r_j} \left\{ q_j v_j - \frac{\beta_t}{2B} q_j \right\},
\]
which represents the maximum utility that transaction $j$ can achieve. Since transactions never accept negative utility, transaction $j$ is scheduled only on blocks whose price does not exceed its value, i.e., $ v_j\geq \frac{\beta_t}{2B}$.

Therefore, the problem of choosing the allocation variables $x_{jt}$ reduces to the problem of setting a price for each block $t$. As a result, transactions with $v_j \ge \frac{\beta_t}{2B}$ are scheduled in block $t$.

Since our objective is to maximize the total value of scheduled transactions, we seek to schedule as many transactions as possible. To achieve this, transactions with high values should be scheduled on blocks with high prices; otherwise, low-value transactions may be left without any blocks on which they can be scheduled. Consequently, we prefer the surplus $\alpha_j$ of each transaction to be as small as possible (ideally zero). At the same time, prices should be kept as low as possible so as to encourage transactions to submit their jobs. This reasoning provides intuition for the dual objective function.

\begin{rem}
If transactions act strategically, they could obtain a reward of at most $\alpha_j$. In this formulation, however, we assume that transactions do not behave strategically. Instead, they purchase the first block whose price is below their value and that has sufficient remaining capacity.
\end{rem}
\begin{rem}
The capacity constraint $
\sum_{j \in \mathcal{C}} q_j x_{jt} \le 2B$ in the primal formulation~\eqref{eq:primal} captures the hard cap on block size imposed by the EIP-1559 mechanism. Note that the target block size~$B$ is not enforced as a feasibility constraint; instead, it emerges endogenously through the price update rule, which adjusts prices in response to block congestion. Our primal--dual formulation therefore enforces feasibility via the hard cap constraint, while regulating the realized block size through price updates. We refer the reader to the supplementary material in Appendix~\ref{app} for a more detailed primal--dual program that directly incorporates the target block size constraint into the objective function of the primal.
\end{rem}

\subsection{Competitive Ratio Analysis of EIP-1559 Using the Primal--Dual Method} 

Building on the linear program relaxation of the static online scheduling problem, in this section, we analyze the competitive ratio of the EIP-1559 algorithm using a principled primal--dual approach.\footnote{The competitive ratio of an online algorithm is the worst-case ratio between the objective value achieved by the online algorithm and that of an optimal offline algorithm with full knowledge of the input.} To that end, we first formally describe the EIP-1559 mechanism in Algorithm~\ref{alg:EIP1559}.
\begin{algorithm}[htb]
\caption{EIP-1559 Pricing Mechanism}
\label{alg:EIP1559}
\begin{algorithmic}[1]
\REQUIRE Target block capacity $B$, hard cap $2B$, update parameter $\eta>0$, minimum price $p_{\min}$, horizon $T$
\STATE Initialize base fee $p_1 \ge p_{\min}$
\FOR{$t=1,2,\dots,T$}
    \STATE Announce base fee $p_t$
    \STATE Include transactions with value $v_j \ge p_t$ until total size\footnotemark $Q_t \le 2B$. 
    \STATE Update price:
    \begin{align}\label{eq:EIP-price-rule}
    p_{t+1} = \max\!\left\{p_{\min},\, p_t \exp\!\left(\eta \frac{Q_t - B}{B}\right)\right\}
    \end{align}
\ENDFOR
\end{algorithmic}
\end{algorithm}
\footnotetext{If the total size of eligible transactions exceeds $2B,$ the miner selects a maximal (by inclusion) subset of eligible transactions whose total size does not exceed the maximum block capacity, prioritizing transactions according to the tips they offer.}

\begin{theorem}\label{theorem:competitive-ratio}
The EIP-1559 algorithm is $\gamma$-competitive for the static online scheduling problem, where
\[
\gamma
=
\frac{2Bv_{\mathrm{max}}}{q_{\text{min}} p_{\text{min}}\left(1+\eta \frac{q_{\text{min}}-B}{B}\right)}+\ \frac{v_{\max}-p_{\min}}{v_{\max}},
\]
with $v_{\mathrm{max}}=\max_j v_j$ denoting the maximum transaction value, $p_{\mathrm{min}}$ the minimum allowable base fee, and $q_{\mathrm{min}}=\min_i q_i$ denoting the minimum transaction size.
\end{theorem}
\begin{proof}
Let $V^{\mathrm{OA}}$ denote the total value of transactions scheduled by the EIP-1559 online algorithm (OA) over the first $T$ blocks, and let $V^{\mathrm{OPT}}$ denote the optimal offline value of scheduled transactions. Let $\mathcal{C}$ denote the set of all transactions, and let $\mathcal{C}_{\mathrm{sched}} \subseteq \mathcal{C}$ denote the subset of transactions scheduled into blocks by the online algorithm. Moreover, let $p_t$ be the EIP-1559 price per unit size in block $t$, and define
\[
\beta_t \coloneqq 2B \cdot p_t \cdot \frac{v_{\mathrm{max}}}{p_{\mathrm{min}}} \ \ \ \ t=1, \ldots, T.
\]
For each transaction $j$, let $\alpha_j$ be any nonnegative number. Then, for every transaction $j \in \mathcal{C}$, we have
\[
q_j v_j \le \frac{\beta_t}{2B} q_j \le \frac{\beta_t}{2B} q_j + \alpha_j,
\]
which implies that $(\alpha,\beta)$ is a feasible solution to the dual program in Eq.~\eqref{eq:dual}. Using weak duality for the primal maximization program in Eq.~\eqref{eq:primal}, we have
\begin{equation}\label{eq:weak-duality}
V^{\mathrm{OPT}}
\le \sum_{t=1}^{T} \beta_t + \sum_{j\in\mathcal{C}} \alpha_j.
\end{equation}

Next, we establish two (independent) lower bounds.
\begin{itemize}
    \item We first derive a lower bound on $V^{\mathrm{OA}}$. For each transaction $j \in \mathcal{C}_{\mathrm{sched}}$, let $t_j$ denote the block in which the online algorithm schedules transaction $j$. Also, define
    \[
    \alpha_j =
    \begin{cases}
    q_jv_j - p_{t_j}q_j, & j \in \mathcal{C}_{\mathrm{sched}},\\
    0, & j \in \mathcal{C}\setminus\mathcal{C}_{\mathrm{sched}}.
    \end{cases}
    \]
    Note that $\alpha_j \ge 0$ for all $j \in \mathcal{C}$. It then follows that
    \begin{align}
    V^{\text{OA}}
    &= \sum_{j\in\mathcal{C}_{\mathrm{sched}}} v_j q_j\notag\\
    &= \sum_{j\in\mathcal{C}_{\mathrm{sched}}} \left(v_j q_j - p_{t_j}q_j + p_{t_j}q_j \right) \notag\\
    &= \sum_{j\in\mathcal{C}_{\mathrm{sched}}} \left(\alpha_j + p_{t_j}q_j\right) \notag\\
    &\ge \sum_{j\in\mathcal{C}} \alpha_j + p_{\mathrm{min}} \sum_{j\in\mathcal{C}_{\mathrm{sched}}} q_j \notag\\
    &\ge \sum_{j\in\mathcal{C}} \alpha_j
       + p_{\mathrm{min}} \sum_{j\in\mathcal{C}_{\mathrm{sched}}}
       \frac{\alpha_j}{v_j- p_{t_j}} \notag\\
    &\ge \left(1+\frac{p_{\text{min}}}{v_{\text{max}}- p_{\text{min}}}\right)
       \sum_{j\in\mathcal{C}} \alpha_j \notag\\
    &= \frac{v_{\text{max}}}{v_{\text{max}}- p_{\text{min}}}
       \sum_{j\in\mathcal{C}} \alpha_j
    \label{eq:alpha-lb}
    \end{align}
    where $p_{\mathrm{min}}$ is the minimum price set by the EIP-1559 protocol and $v_\mathrm{max}$ is the maximum value of the transactions. Note that we assume $v_j \neq p_t$ for every transaction $j \in \mathcal{C}$ and every block $t=1,\ldots,T$, and also $p_{t}< v_{\mathrm{max}}$.

    \item We now derive a lower bound for $V^{\mathrm{OA}}$ based on $\sum_{t=1}^{T} \beta_t$. For this bound, we assume that there are sufficiently many transactions so that EIP-1559 schedules at least $q_{\mathrm{min}}$ units in every block. Let $\mathrm{Alg}^t$ denote the total value of transactions scheduled by the online algorithm over the first $t$ blocks, so that $\mathrm{Alg}^{t+1} - \mathrm{Alg}^{t}$ represents the value scheduled in block $t+1$ by the online algorithm. Also, let $Q_t = \sum_{j \in \mathcal{C}} q_j x_{jt}$ be the total size of transactions scheduled in block $t$, and let $\mathcal{C}_{\mathrm{sched}}^t$ denote the set of transactions scheduled into block $t$ by the online algorithm. We then have
    \begin{align*}
        \text{Alg}^{t+1} - \text{Alg}^{t} &= \sum_{j\in \mathcal{C}_{\mathrm{sched}}^{t+1}} q_j v_j\cr 
        &\geq p_{t+1} \sum_{j\in\mathcal{C}_{\mathrm{sched}}^{t+1}} q_j\cr 
        &= p_{t+1} Q_{t+1}\cr 
        &\geq p_{t+1} q_{\text{min}},
    \end{align*}
    where $p_{t+1}$ is the price of block $t+1$ and is given by $p_{t+1} = \text{max} \{p_{\text{min}}, p_t e^{\eta(Q_t-B)/B}\}$. Using the fact that $p_t = \frac{\beta_t p_{\mathrm{min}}}{2Bv_{\mathrm{max}}}$, we have
    \begin{align*}
        \text{Alg}^{t+1} - \text{Alg}^{t}
        &\geq q_{\text{min}} p_t e^{\eta(Q_t-B)/B} \cr
        &\geq q_{\text{min}} \frac{\beta_t p_{\mathrm{min}}}{2Bv_{\mathrm{max}}} e^{\eta(q_{\text{min}}-B)/B}\cr
        &\geq q_{\text{min}} \frac{\beta_t p_{\mathrm{min}}}{2Bv_{\mathrm{max}}} \left(1+\eta \frac{q_{\text{min}}-B}{B}\right). 
    \end{align*}
    Given that the online algorithm can run up to time $T$, we have
    \begin{align}
        V^{\text{OA}} &= \text{Alg}^{T} \\&=
        \text{Alg}^{1} + \sum_{t=1}^{T-1} \text{Alg}^{t+1} - \text{Alg}^{t}
        \\&\geq \frac{q_{\text{min}} p_{\text{min}}\left(1+\eta \frac{q_{\text{min}}-B}{B}\right)}{2Bv_{\mathrm{max}}} \sum_{t=1}^{T-1} \beta_t \\&\geq \Bigg(\frac{q_{\text{min}} p_{\text{min}}\left(1+\eta \frac{q_{\text{min}}-B}{B}\right)}{2Bv_{\mathrm{max}}}\Bigg) \left(\frac{T-1}{T-1+\frac{v_{\mathrm{max}}}{p_{\mathrm{min}}}}\right)\sum_{t=1}^{T} \beta_t.
        \label{eq:beta-lb}
    \end{align}
    \end{itemize}

To complete the proof, by combining  the lower bounds in Eq.~\eqref{eq:alpha-lb} and Eq.~\eqref{eq:beta-lb}, we get
\[
\sum_{t=1}^{T}\beta_t+\sum_{j\in\mathcal{C}}\alpha_j
\ \le\
\left(\Bigg(\frac{2Bv_{\mathrm{max}}}{q_{\text{min}} p_{\text{min}}\left(1+\eta \frac{q_{\text{min}}-B}{B}\right)}\Bigg)\Bigg(\frac{T-1+\frac{v_{\mathrm{max}}}{p_{\mathrm{min}}}}{T-1}\Bigg)+\ \frac{v_{\max}-p_{\min}}{v_{\max}}\right)\ V^{\mathrm{OA}}.
\]
This in view of ~\eqref{eq:weak-duality} shows that the EIP-1559 algorithm is $\gamma$-competitive with
\begin{equation}\label{eq:gamma}
\gamma
 =
\Bigg(\frac{2Bv_{\mathrm{max}}}{q_{\text{min}} p_{\text{min}}\left(1+\eta \frac{q_{\text{min}}-B}{B}\right)}\Bigg)\Bigg(\frac{T-1+\frac{v_{\mathrm{max}}}{p_{\mathrm{min}}}}{T-1}\Bigg)+\ \frac{v_{\max}-p_{\min}}{v_{\max}}.
\end{equation}
Note that as $T\rightarrow\infty$, the competitive ratio converges to $\gamma
 = \frac{2Bv_{\mathrm{max}}}{q_{\text{min}} p_{\text{min}}\left(1+\eta \frac{q_{\text{min}}-B}{B}\right)}+\ \frac{v_{\max}-p_{\min}}{v_{\max}}$.
\end{proof}

\begin{rem}
    While the competitive ratio of EIP-1559 was analyzed in \cite{babaioff2024optimality}, our analysis has several key advantages: (i) it provides a simple and principled analysis based on the primal--dual method; (ii) unlike the competitive-ratio analysis in \cite{babaioff2024optimality}, we do not consider additional $\Gamma$ time steps for the online algorithm; and (iii) it offers an interpretable way of choosing prices in terms of the dual variables of the underlying linear program, which in turn allows us to extend this method to a dynamic setting. 
\end{rem}
\begin{rem}
The competitive ratio established in Theorem~\ref{theorem:competitive-ratio} explicitly incorporates the EIP-1559 price update rule. Alternatively, one can derive a smaller competitive ratio by using only the fact that EIP-1559 schedules at least $q_{\mathrm{min}}$ units in every block, without explicitly exploiting the price update dynamics. For completeness, we present this alternative analysis in Appendix~\ref{appx:alternative-competitive-ratio}.
\end{rem}

\section{An MDP Formulation for the Dynamic Scheduling Problem}\label{sec:dynamic}

The preceding primal--dual analysis motivates the use of prices to regulate transaction inclusion under block capacity constraints. We now study how these prices should depend on the evolving mempool. In this section, we formulate the dynamic scheduling problem as an MDP whose objective balances mempool congestion against scheduling beyond the target block capacity.

We consider a pool of transactions with \emph{heterogeneous} sizes and values.\footnote{We refer to Section~\ref{sec:threshold} for an equivalent simplified version of this model in the case of homogeneous transactions.} Transactions take per-unit values from the finite set $\mathcal V = \{v_1,\dots,v_n\}$ with $n \geq 2$, and sizes from the finite set $\mathcal Q = \{q_1,\dots,q_m\}$. At each time $t = 1,2,\ldots$, the system state is represented by a matrix $\mathbf{S}_t \in \mathbb{Z}_{\ge 0}^{m \times n}$, where the $(i,j)$th entry, denoted by $\mathbf{S}_t[i,j]$, records the number of transactions of size $q_i$ and value $v_j$ waiting in the pool at time $t$. The action at time $t$ is an index $a_t \in \{1,\dots,n\}$, corresponding to a value threshold $v_{a_t} \in \mathcal V$. All transactions with value at least $v_{a_t}$ are \emph{eligible} to be scheduled on block $t$. 

Let $\mathcal{F}$ be a scheduling protocol that is determined at the beginning and fixed throughout the process, such that for any state $\mathbf{S}$ and action $a$, $\mathcal{F}(\mathbf{S}, a)$ determines how the eligible transactions at that state (i.e., those whose value exceeds $v_a$) are packed into a block of size at most $2B$. Accordingly, $\mathcal{F}(\mathbf{S}, a)$ can be represented as an $(m,n)$-dimensional matrix, where the $(i,j)$th entry, denoted by $\mathcal{F}(\mathbf{S}, a)[i,j]$, specifies the number of transactions of size $q_i$ and value $v_j$ scheduled into the block under state $\mathbf{S}$ and action $a$. Given the scheduling protocol $\mathcal{F}$, we let $Q_{\mathrm{sched}}^{\mathcal{F}}(\mathbf{S}, a)$ denote the total size of transactions scheduled under state $\mathbf{S}$ and action $a$ when following protocol $\mathcal{F}$, i.e., 
\[
Q_{\mathrm{sched}}^{\mathcal{F}}(\mathbf{S},a) = \sum_{i=1}^m q_i \sum_{j=1}^n \mathcal{F}(\mathbf{S},a)[i,j].
\]

\begin{rem}\label{rem:schedul-F}
    While $\mathcal{F}$ may take a general form, for concreteness we adopt in this work the following natural choice, which follows a bang-per-buck packing rule up to some threshold $\tau$.\footnote{In fact, we will show in Section~\ref{sec:threshold} that this class of scheduling protocols is optimal in the case of homogeneous transactions.} Specifically, given a state--action pair $(\mathbf{S}, a)$, the protocol first sorts all eligible transactions in state $\mathbf{S}$ in ascending order of the product of their value and size, i.e., $v \cdot q$, and then packs them into the block until a maximum capacity threshold $\tau\leq 2B$ is reached. To ensure that the state space remains finite, whenever $\mathbf{S}[i,j] > L$ for some $(i,j)$, the protocol first schedules transactions from $\mathbf{S}[i,j]$ until it is reduced to $L$. Any remaining block capacity is then allocated according to the original sorting rule.
\end{rem}
Let $\mathbf{A}_t$ be a stationary stochastic arrival matrix, where $\mathbf{A}_t[i,j]$ is a random variable denoting the number of new transactions of size $q_i$ and value $v_j$ that arrive in the system at time $t$. The state evolution of the system is then governed by the following dynamics:
\begin{align}
\label{eq:state-dynamics}
    \mathbf{S}_{t+1} = \mathbf{S}_{t} + \mathbf{A}_t - \mathcal{F}(\mathbf{S}_t,a_t),
\end{align}
and the state transition probability function is defined accordingly as:
\begin{align}\label{eq:prob-tran-matrix}
    P(\mathbf{S}_{t+1} | \mathbf{S}_{t}, a_t, \mathcal{F}) = \mathbb{P}(\mathbf{A}_t =  \mathbf{S}_{t+1} - \mathbf{S}_{t} + \mathcal{F}(\mathbf{S}_t,a_t)). 
\end{align}
Since the scheduling protocol $\mathcal{F}$ is assumed to be fixed, we can suppress it in the notation and write the transition probability as $P(\mathbf{S}_{t+1} \mid \mathbf{S}_t, a_t)$. After the transition occurs, the reward received is
\begin{equation}
\label{reward}
    \begin{split}
    r(\mathbf{S}_t,a_t) &= -c_{\mathrm{hold}} Q^t_{\mathrm{pool}} - c_{\mathrm{over}}  \left(Q_{\mathrm{sched}}^{\mathcal{F}}(\mathbf{S}_t,a_t)- B\right)^{+},
    \end{split}
\end{equation}
where $(\cdot)^+=\max\{0,\cdot\}$, and $Q^t_{\mathrm{pool}}$ is the total size of transactions in the pool at time $t$, defined as
\begin{align}\label{eq:Q-t}
Q^t_{\mathrm{pool}} = \sum_{i=1}^m q_i \sum_{j=1}^n \mathbf{S}_t[i,j].
\end{align}
The parameter $c_{\mathrm{hold}}$ denotes the holding cost incurred by transactions that remain in the pool, whereas $c_{\mathrm{over}}$ represents the penalty for scheduling transactions in excess of the target block size~$B$. Finally, our goal is to maximize the long-run discounted average reward, defined as
\[
V^\pi(\mathbf{S}_0) = \mathbb{E}_\pi [\sum_{t=1}^{\infty} \gamma^t r(\mathbf{S}_t, a_t)],
\]
where $\pi$ is the \emph{stationary policy} followed by the mechanism, mapping each state to a probability distribution over actions, and $\gamma \in (0, 1)$ is the discount factor. We further assume that the initial state $\mathbf{S}_0$ is the all-zero matrix, indicating that the pool is initially empty.

Note that the dynamic scheduling problem involves a tradeoff between pool congestion and block-capacity overshoot. Scheduling more transactions reduces the congestion and future holding costs, but may incur an overshoot penalty when the scheduled volume exceeds the target block capacity $B$. Conversely, limiting the scheduled volume avoids this penalty but leaves more transactions waiting in the pool. The price controls transaction eligibility and thereby influences this tradeoff through the fixed scheduling protocol $\mathcal{F}$. Our objective is to choose a state-dependent pricing policy that maximizes the long-run discounted average reward, balancing mempool congestion and block-capacity overshoot subject to the hard block capacity constraint $2B$.

\begin{rem}
    For stationary transaction arrivals, finding the optimal policy is equivalent to solving a finite-state, finite-action MDP, for which the existence of a deterministic stationary optimal policy is well established \cite{abbasi2013online}. Therefore, in the remainder of the paper, without loss of generality, we restrict our analysis to the class of stationary deterministic policies, i.e., mappings $\pi: \mathcal{S} \to \mathcal{A}$ that assign a single action to each state independently of the time step $t$.  
\end{rem}

\begin{rem}
    The structure of the reward function \eqref{reward} induces different scheduling priorities depending on the relative magnitudes of the holding cost $c_{\mathrm{hold}}$ and the overshoot penalty $c_{\mathrm{over}}$. When $c_{\mathrm{hold}} \gg c_{\mathrm{over}}$, maximizing the reward favors scheduling all incoming transactions as quickly as possible, keeping the mempool nearly empty even at the expense of exceeding the target block size. In contrast, when $c_{\mathrm{hold}} \ll c_{\mathrm{over}}$, the optimal policy emphasizes adhering to the target block capacity $B$, potentially allowing the mempool to become crowded in order to avoid overshoot penalties.
\end{rem}

\begin{rem}
\label{B-Q-range}
   Let $Q^t_{\mathrm{arrival}}$ denote the total size of transactions arriving at time $t$. If $Q^t_{\mathrm{arrival}} > 2B$ for all $t$, it becomes impossible to schedule all arriving transactions, and the system will eventually become unstable. Accordingly, in the stochastic arrival setting, we assume that $Q^t_{\mathrm{arrival}}$ is a random variable taking values in the range $[0, 2B]$. In the deterministic setting, we assume $B < Q_{\mathrm{arrival}} \leq 2B$, since if $Q_{\mathrm{arrival}} \leq B$, all transactions can be scheduled within a single block, making the problem trivial.
\end{rem}

In our formulation, we assume that the protocol has complete information about the number of transactions of each type in the mempool. However, we note that our framework can be extended to settings in which the protocol has only partial observations of the mempool, thanks to the extension of MDPs to partially observable Markov decision processes (POMDPs). For instance, if the protocol can observe only the history of included transactions, rather than the entire mempool, it can use the history of scheduled transactions to form and update a belief distribution over the mempool (state), which becomes more accurate over time. In this case, the system can be modeled as a POMDP with belief states rather than an MDP with fully observable states, and existing results \cite{anjarlekar2026scalable} can be used to reduce it to an MDP with a small error in the optimal policy evaluation. In our setting, however, as an initial model and to keep the derivations simple, we assume full information and that the distribution of transactions in the mempool is known.

\section{An Online Algorithm for the Dynamic Scheduling Problem}

In this section, we provide an online algorithm for obtaining the optimal policy that maximizes the long-run discounted average reward under the given state dynamics.

Recall that the long-run discounted average reward of a policy \(\pi\), starting from the initial state \(\mathbf{S}_0\), is defined as
\begin{align*}
    V^\pi(\mathbf{S}_0)
    &=
    \mathbb{E}_\pi
    \left[
    \sum_{t=1}^{\infty} \gamma ^t r(\mathbf{S}_t, a_t) \mid\mathbf{S}_0
    \right].
\end{align*}

To maximize this objective in the dynamic scheduling problem in an online fashion, we employ the Natural Policy Gradient (NPG) algorithm \cite{even2009online}, adapted to our setting. The NPG algorithm defines a Fisher information matrix and performs gradient updates in the geometry induced by this matrix as follows:
\begin{align}
    F_{\rho}(\theta)
    &=
    \mathbb{E}_{s \sim \nu_\rho^{\pi_\theta}}
    \mathbb{E}_{a \sim \pi_\theta(\cdot \mid s)}
    \left[
    \nabla_\theta \log \pi_\theta(a \mid s)
    \bigl(\nabla_\theta \log \pi_\theta(a \mid s)\bigr)^{\top}
    \right], \label{eq:NPG-update}\\
    \theta^{(t+1)}
    &=
    \theta^{(t)}
    +
    \eta\, F_{\rho}(\theta^{(t)})^{\dagger}
    \nabla_\theta V^{(t)}(\rho) \notag,
\end{align}
where $\{\pi_\theta\mid\theta \in \Theta\}$ is the class of stationary policies parametrized by $\theta$, $\eta$ is the learning rate, $\rho$ is the initial state distribution, and $\nu^{\pi_\theta}_\rho$ is the state visitation distribution under policy $\pi_\theta$ and $\rho$. We use softmax parametrization, where for unconstrained $\theta \in \mathbb{R}^{\mid\mathcal{S}\mid \mid\mathcal{A}\mid}$, the policy is parametrized as
\begin{align}
\label{eq:softmax-param}
    \pi_{\theta}(a\mid \mathbf{S}) = \frac{\mathrm{exp}\left(\theta_{\mathbf{S},a}\right)}{\sum_{a' \in \mathcal{A}} \mathrm{exp}\left(\theta_{\mathbf{S},a'}\right)}.
\end{align}
Next, we state the following lemma from \cite[Lemma~15 and Theorem~16]{JMLR:v22:19-736}, which provides a performance guarantee for the convergence of the NPG algorithm to an approximately optimal stationary policy.
\begin{lemma}
    For the softmax parameterization~\eqref{eq:softmax-param}, the NPG updates in~\eqref{eq:NPG-update} take the form
    \begin{align*}
        \pi^{(t+1)}(a \mid s) &= \pi^{(t)}(a \mid s)
        \frac{\exp\left(\eta A^{(t)}(s,a)/(1-\gamma)\right)}{Z_t(s)},
    \end{align*}
    where $A^{(t)}(s,a):= Q^{\pi^{(t)}}(s,a) - V^{\pi^{(t)}}(s)$ is the advantage function for policy $\pi^{(t)}$, and 
    \[
    Z_t(s) = \sum_{a' \in \mathcal{A}} \pi^{(t)}(a' \mid s)
    \exp\left(\eta A^{(t)}(s,a')/(1-\gamma)\right).
    \]
    In particular, by setting 
    \begin{align}
    \label{eq:lr-condition}
        \eta \geq (1-\gamma)^2 \log|\mathcal{A}|,
    \end{align}
    the NPG finds an $\epsilon$-optimal policy in at most $T \leq \frac{2}{(1-\gamma)^2\epsilon}$ iterations.
\end{lemma}

\begin{rem}
In practice, the advantage function is unknown and is estimated from trajectories of finite length, which introduces sampling and truncation errors. However, the same results can be extended to the setting where access to the exact advantage is replaced by an estimation oracle. For ease of presentation and to avoid additional technical details associated with estimation errors, we assume access to an oracle for the exact advantage and instead evaluate the performance of the sample-based implementation numerically in Section~\ref{sec:numerical-results}, where the advantage is estimated from the online trajectory.
\end{rem}

\begin{algorithm}[htb]
\caption{Episodic NPG for Dynamic Pricing of the Ethereum Mempool}
\begin{algorithmic}
\label{alg:NPG}
\STATE \textbf{Input:} transaction size set $\mathcal{Q} \!=\! \{q_1,\dots,q_m\}$, value per size set $\mathcal{V} \!=\! \{v_1,\dots,v_n\}$, target block size $B$, scheduling protocol $\mathcal{F}$, learning rate $\eta$, discount factor $\gamma$, horizon $H$, number of iterations $T$.
\STATE \textbf{Initialize:} environment $\mathcal{E}$, initial state $\mathbf{S}_0 \leftarrow \textsc{Reset}(\mathcal{E})$, initial policy $\pi^{(0)}(\cdot \mid \mathbf{S}_0)$ over $n$ actions.
\FOR{iteration $k = 0$ to $T-1$}
    \STATE Initialize episode buffer $\mathcal{B}_k \leftarrow \emptyset$.
    \STATE Set $\mathbf{S}_0^{(k)} \leftarrow \textsc{Reset}(\mathcal{E})$.
    \FOR{$t = 0$ to $H-1$}
        \STATE Sample $a^{(k)}_t \sim \pi^{(k)}(\cdot \mid \mathbf{S}^{(k)}_t)$.
        \STATE Apply the scheduling protocol $\mathcal{F}$ to obtain
        $\mathcal{F}(\mathbf{S}_t,a_t) \in \mathbb{Z}_+^{m \times n}$,
        where the $(i,j)$-th element denotes the number of transactions scheduled with size $q_i$ and value per size $v_j$.
        \STATE Calculate the reward $r_t^k$ using Eq.~\eqref{reward}.
        \STATE Calculate the next state $\mathbf{S}_{t+1}^{(k)}$ using Eq.~\eqref{eq:state-dynamics}.
        \STATE Step environment: $(r^{(k)}_t, \mathbf{S}^{(k)}_{t+1}) \leftarrow \textsc{Step}(\mathcal{E}, a^{(k)}_t)$.
        \STATE Store transition:
        \STATE \hspace{1.2em} $\mathcal{B}_k \leftarrow \mathcal{B}_k \cup \{(\mathbf{S}^{(k)}_t,a^{(k)}_t,r^{(k)}_t,\mathbf{S}^{(k)}_{t+1})\}$.
    \ENDFOR
    \STATE Compute advantage  $A^{(k)}(\mathbf{S},a)$ for the current policy $\pi^{k}$.
    \FOR{each state $\mathbf{S}$ for which $(\mathbf{S},\cdot)$ appears in $\mathcal{B}_k$}
        \STATE Compute the normalization constant:
        \STATE \hspace{1.2em} $Z_{k}(\mathbf{S}) \leftarrow \sum_{a'} \pi^{(k)}(a'\mid \mathbf{S})\exp\left(\frac{\eta A^{(k)}(\mathbf{S},a')}{1-\gamma}\right)$.
        \FOR{each action $a \in \{1,\dots,n\}$}
            \STATE Update: 
            \STATE \hspace{1.2em} $\pi^{(k+1)}(a\mid \mathbf{S}) =
            \frac{\pi^{(k)}(a\mid \mathbf{S})\exp\left(\frac{\eta A^{(k)}(\mathbf{S},a)}{1-\gamma}\right)}{Z_{k}(\mathbf{S})}$.
        \ENDFOR
    \ENDFOR
\ENDFOR
\end{algorithmic}
\end{algorithm}

\subsection{EIP-1559 as a Special Case of a Natural Policy Gradient Algorithm}
As discussed earlier, the EIP-1559 mechanism does not account for the dynamic evolution of the transaction mempool in its price update rule. We therefore examine whether the static version of Algorithm~\ref{alg:NPG} has an update rule analogous to that of EIP-1559.

\begin{proposition}
Suppose the mempool is sufficiently congested such that adding or removing a small number of transactions does not affect its state; that is, the system reaches a stationary regime with $\mathbf{S}_t = \mathbf{S}$ for all $t$. Then, the dynamics of Algorithm~\ref{alg:NPG} resemble the EIP-1559 price update rule~\eqref{eq:EIP-price-rule}.
\end{proposition}
\begin{proof}
Under the stationary assumption, the system resembles a stateless, static setting, and the $Q$-function under policy $\pi$ can be written as
\begin{align*}
    Q^\pi(\mathbf{S},a) &= 
    r(\mathbf{S},a) + \gamma V^\pi(\mathbf{S}).
\end{align*}
We can write the value functions as
\begin{align*}
    V^\pi(\mathbf{S})
    = \mathbb{E}_{a \sim \pi(\cdot \mid \mathbf{S})}\left[
    Q^\pi(\mathbf{S},a) \right]
    = \mathbb{E}_{a \sim \pi(\cdot \mid \mathbf{S})}\left[
    r(\mathbf{S},a) + \gamma V^\pi(\mathbf{S})\right],
\end{align*}
which implies
\begin{align*}
    (1-\gamma)\,V^\pi(\mathbf{S})
    =
    \mathbb{E}_{a \sim \pi(\cdot \mid \mathbf{S})}\left[
    r(\mathbf{S},a)
    \right].
\end{align*}
The advantage function is defined as
\begin{align*}
    A^\pi(\mathbf{S},a) = Q^\pi(\mathbf{S},a) - V^\pi(\mathbf{S}).
\end{align*}
Substituting the expressions above yields,
\begin{align}
    A^\pi(\mathbf{S},a)
    &= r(\mathbf{S},a) - (1-\gamma)V^\pi(\mathbf{S}) \notag\\
    &= r(\mathbf{S},a)
    - \mathbb{E}_{a' \sim \pi(\cdot \mid \mathbf{S})}\left[
    r(\mathbf{S},a')
    \right]. 
    \label{eq:adv-fixed-state}
\end{align}
Since the state is static, i.e., $\mathbf{S}^t=\mathbf{S}$ for all $t$, the total size of the transactions in the mempool given by~\eqref{eq:Q-t} is fixed and constant. Consequently, this term is canceled due to the normalization in the NPG update rule. Therefore, without loss of generality, we may set the holding cost to zero, i.e., $c_{\mathrm{hold}}=0$, in the reward function~\eqref{reward}. Now, by substituting the advantage function in Eq.~\eqref{eq:adv-fixed-state} in the policy update in Algorithm~\ref{alg:NPG}, we get
\begin{align}\nonumber
        \pi^{(t+1)}(a\mid \mathbf{S}) &= \frac{\pi^{(t)}(a\mid \mathbf{S})\exp\left(\frac{\eta A^{(t)}(\mathbf{S},a)}{1-\gamma}\right)}{\sum_{\hat{a}} \pi^{(t)}(\hat{a}\mid \mathbf{S})\exp\left(\frac{\eta A^{(t)}(\mathbf{S},\hat{a})}{1-\gamma}\right)}\cr 
        &\simeq \frac{\pi^{(t)}(a\mid \mathbf{S})\exp\left(\frac{\eta \left(r(\mathbf{S},a) - \mathbb{E}_{a' \sim \pi^{(t)}}\left[r(\mathbf{S},a')\right]\right)}{1-\gamma}\right) }{\sum_{\hat{a}} \pi^{(t)}(\hat{a}\mid \mathbf{S})\exp\left(\frac{\eta \left(r(\mathbf{S},\hat{a}) - \mathbb{E}_{a' \sim \pi^{(t)}}\left[r(\mathbf{S},a')\right]\right)}{1-\gamma}\right)}\cr 
        &= 
        \frac{\pi^{(t)}(a\mid \mathbf{S})\exp\left(\frac{\eta r(\mathbf{S},a)}{1-\gamma}\right) }{\sum_{\hat{a}} \pi^{(t)}(\hat{a}\mid \mathbf{S})\exp\left(\frac{\eta r(\mathbf{S},\hat{a})}{1-\gamma}\right)}\nonumber
        \\&\propto \pi^{(t)}(a\mid \mathbf{S})\exp\left(\frac{\eta r(\mathbf{S},a)}{1-\gamma}\right)\cr 
        &=
        \pi^{(t)}(a\mid \mathbf{S})\exp\Bigg(\frac{- \eta c_{\mathrm{over}}  \left(Q_{\mathrm{sched}}^{\mathcal{F}}(\mathbf{S},a)- B\right)^{+}}{1-\gamma}\Bigg),
\end{align}
which resembles the EIP-1559 price update rule \eqref{eq:EIP-price-rule} up to scaling factors in which $p_t\sim \pi^{(t)}(a^t|\mathbf{S})$ and $Q_t\sim Q_{\mathrm{sched}}^{\mathcal{F}}(\mathbf{S},a^t)$.  
\end{proof}

In Section~\ref{sec:numerical-results}, we show that when $c_{\mathrm{over}} \gg c_{\mathrm{hold}}$, the average amount of scheduled transactions converges to the target block capacity $B$, consistent with the results shown for EIP-1559 \cite{babaioff2024optimality}.

\section{Optimality of Threshold Policies for Homogeneous Dynamic Scheduling Problem}\label{sec:threshold}

While the NPG Algorithm \ref{alg:NPG} provides an efficient method to obtain an optimal policy for the dynamic scheduling problem, the structure of the optimal policy can be quite complex and depend on many parameters involved in the problem. However, it turns out that by restricting the problem to certain special cases, one can obtain a closed-form characterization of the optimal scheduling policy, which provides useful insights into the behavior of the optimal policy in general. Therefore, in this section we consider a special case of the MDP introduced in Section~\ref{sec:dynamic}, in which all transactions are homogeneous, i.e., they have the same value and size. In this setting, pricing becomes irrelevant, since no transaction is preferred over another. As a result, the \emph{homogeneous} dynamic scheduling problem reduces to an equivalent formulation in which the protocol directly controls the total volume of scheduled transactions through the scheduling protocol $\mathcal{F}$, as described in detail below. Note that unlike the preceding formulation, which fixes the scheduling protocol and optimizes the price, here we make all transactions eligible and optimize the scheduled volume directly. We also keep the same long-run discounted average reward objective, so this homogeneous setting allows us to characterize how the optimal scheduling policy balances the tradeoff between mempool congestion and block-capacity overshoot.

\subsection{Reformulation of the Dynamic Scheduling Problem under Homogeneous Arrivals}

Let $s_t$ denote the state of the system at time $t$, represented by a scalar quantity corresponding to the total size of transactions in the mempool at time $t$. Let $A_t$ denote the total size of transactions arriving at time $t$, where $\{A_t\} \sim A$ is an i.i.d.\ process with $A \in [0, 2B]$. Since setting a price action $a_t$ either makes all transactions eligible or ineligible (depending on whether $v_{a_t}$ is smaller or larger than the common transaction value $v$, respectively), the scheduler should clearly set a low price to ensure that all transactions are eligible; otherwise, no transactions can be scheduled. In this case, the only remaining degree of freedom for the scheduler is the choice of scheduling protocol $\mathcal{F}$.\footnote{Note that in Section~\ref{sec:dynamic}, we fixed the scheduling protocol $\mathcal{F}$ to the one described in Remark~\ref{rem:schedul-F}. Here, we show that this structure is indeed optimal in the homogeneous setting.} 

In order to characterize the optimal scheduling protocol in the homogeneous setting, let us define the action $f_t$ as the amount of overshoot above the target capacity $B$ scheduled at time $t$. By feasibility constraints, it must hold that\footnote{For simplicity of the analysis, we assume that $f_t$ can take fractional values, which is realistic and becomes more accurate when transaction sizes are much smaller than the target block capacity $B$.}
\[
    -B \leq f_t \leq \min\{s_t, 2B\} - B.
\]
The state then evolves as 
\[
s_{t+1} = s_t+A_t-\left(f_t+B\right).
\]
The \emph{instantaneous cost} at time $t$ is given by 
\[
c(s_t, f_t)=c_{\mathrm{hold}} s_t + c_{\mathrm{over}} (f_t)^{+},
\]
which is the negative of the instantaneous reward function defined in Eq.~\eqref{reward}. Note that this cost is bounded since we have imposed a cap of $L$ on the state, i.e., there exists $L > 0$ such that $s_t \leq L\ \forall t$. For the remainder of this section, we assume that $c_{\mathrm{over}} > c_{\mathrm{hold}}$, as the analysis for the other case is similar (see Remark~\ref{rem:trivial-hold} below). 

As before, we consider the infinite horizon discounted cost setting, i.e.,
\begin{equation}
\label{J-infinite-horizon}
    J(s) = c_{\mathrm{hold}} s + \min_{-B \leq f\leq \min\{s,2B\}-B} \left\{c_{\mathrm{over}} (f)^{+} + \gamma \mathbb{E}[ J(s')] \right\},
\end{equation}
where $s' = s + A - (f + B)$, and the expectation is taken over random arrivals $A$. Now, if we define $$Q(s,f) = c_{\mathrm{hold}} s + c_{\mathrm{over}} (f)^{+} + \gamma \mathbb{E}[ J(s') ],$$ we can rewrite \eqref{J-infinite-horizon} equivalently as
\begin{equation}
\label{Q-definition}
    J(s) = \min_{-B \leq f\leq \min\{s,2B\}-B} \left[Q(s,f)\right].
\end{equation}

\subsection{Characterization of the Optimal Scheduling Policy}

In order to characterize the structure of the optimal scheduling protocol for the homogeneous dynamic scheduling problem described above, we apply a value iteration method to compute the optimal cost function $J^*(\cdot)$. To that end, let $J_0(s) = c_{\mathrm{term}} s \ \forall s$. The value iteration update equals
\begin{equation}
\label{J-finite-horizon}
    J_{k+1}(s) = c_{\mathrm{hold}} s 
    + \min_{-B \leq f \leq \min\{s,2B\}-B} 
    \left\{ c_{\mathrm{over}} (f)^{+} 
    + \gamma \mathbb{E}\left[ J_k(s') \right] \right\},
\end{equation}
which is iterated until convergence to the fixed point. Since the instantaneous cost is bounded and $\gamma \in (0,1)$, value iteration converges to the optimal cost function \cite{bertsekas2012dynamic}, i.e.,
\[
    J^*(s) = \lim_{k\to\infty} J_k(s).
\]
Next, we establish the following two properties of the optimal cost function. The proofs of these lemmas are deferred to Appendix~\ref{proof-J-k-convex} and Appendix~\ref{proof-J-non-decreasing}, respectively. 

\begin{lemma}
\label{J-k-convex}
    $J^*(\cdot)$ is a convex function. Moreover, if we let $J^*(s) = \min_{-B \leq f \leq \min\{s,2B\}-B} \left[Q^*(s,f)\right]$, then for any $s \geq 0$, the function $Q^*(s,\cdot)$ is convex in its second argument.    
\end{lemma}

\begin{lemma}
\label{J-non-decreasing}
    $J^*$ is non-decreasing. 
\end{lemma}

Finally, by combining Lemmas~\ref{J-k-convex} and~\ref{J-non-decreasing}, we show that the optimal policy is a threshold policy.
\begin{theorem}
\label{threshold-policy}
There exists a threshold $s^* \geq B$ such that the optimal policy $f^*$ is given by
    \begin{align*}
        f^*(s) =
        \begin{cases}
        \min\{s-B, 0\}, & 0 \leq s \leq s^*, \\[6pt]
        \min\{s - s^*, B\}, & s > s^*.
        \end{cases}
    \end{align*}
\end{theorem}
\begin{proof}
Define $y:= s-(f+B)$ to be the pool size after the scheduling happens. We can therefore rewrite Eq.~\eqref{J-infinite-horizon} for the optimal cost function $J^*$ as 
\[
J^*(s) = c_{\mathrm{hold}}s + \min_{\max\{0, s-2B\} \leq y \leq s} \{c_{\mathrm{over}}(s-y-B)^+ + \gamma \mathbb{E}[J^*(y+A)]\}.
\]
Let $g(y) := \gamma \mathbb{E}[J^*(y + A)]$. Since $J^*$ is convex and non-decreasing by Lemmas ~\ref{J-k-convex} and~\ref{J-non-decreasing}, it follows that $g$ is also convex and non-decreasing. Thus, we have
\begin{align*}
    J^*(s) &= 
     \min_{\max\{0, s-2B\} \leq y \leq s} \{c_{\mathrm{hold}}s + c_{\mathrm{over}}(s-y-B)^+ + g(y)\} \\&=
     \min_{\max\{0, s-2B\} \leq y \leq s} [Q^*(s,y)],
\end{align*}
which, in view of Lemma~\ref{J-k-convex}, shows that $Q^*(s,y)$ is convex in $y$ for any $s$. Now we have three cases:
\begin{itemize}
    \item $s \leq B$: In this case, we have $J^*(s) = c_{\mathrm{hold}} s + \min_{0 \leq y \leq s} \{ g(y) \}$, and since $g$ is non-decreasing, the optimal solution is attained at $y^*(s)=0$, and thus $f^*(s)=s-B$.
    \item $B < s \leq y+B$: In this case, although $s>B$, no overshoot occurs. Therefore $y^*(s)+B = s$ and $f^*(s) = 0$.
    \item $y+B\leq s$: In this case, we have $Q^*(s,y)=c_{\mathrm{hold}}s+c_{\mathrm{over}}(s-B)+(g(y)-c_{\mathrm{over}}y)$ and therefore
    \[
    J^*(s) = c_{\mathrm{hold}}s+c_{\mathrm{over}}(s-B) +\min_{\max\{0, s-2B\} \leq y \leq s-B} [g(y) - c_{\mathrm{over}}y].
    \]
    Since $g$ is convex, $g(y)-c_{\mathrm{over}}y$ is also convex. Suppose $x^* = \arg\min_{y\in[0, \infty)} \{g(y) -c_{\mathrm{over}}y\}$. If $x^*\in[\max\{0, s-2B\}, s-B]$, then $y^*(s)=x^*$. If $x^* > s-B$, then $y^*(s)=s-B$, and if $x^*<\max\{0, s-2B\}$, then $y^*(s) = \max\{0, s-2B\}$. In the case of $g(y) - c_{\mathrm{over}}y$ not having a minimizer on $[0, \infty)$, the minimum happens at one of the two endpoints of the interval and we have $x^*=y^*(s) = \arg\min_{y\in\{\max\{0, s-2B\}, s-B\}} \{g(y)-c_{\mathrm{over}}y\}$. Thus, in any case, $y^*(s)$ is the projection of $x^*$ to the interval $[\max\{0, s-2B\}, s-B]$, which can be written as
    \begin{align*}
    y^*(s) &= \min\{\max\{x^*, \max\{0, s - 2B\}\},\, s - B\}.
    \end{align*}
    Subsequently, $f^*(s) = s-y^*(s)-B$, and we have
    \begin{align}\label{eq:f-structure}
        f^*(s) &= s - \min\{\max\{x^*, \max\{0, s - 2B\}\},\, s - B\} - B \cr&=
        \max\{s-B-\max\{x^*, \max\{0, s - 2B\}\}, 0\} \cr&= 
        \max\{\min\{s-B-x^*, \min\{s-B, B\}\}, 0\} \cr&= \min\{\max\{s-(B+x^*),0\},B\}.
    \end{align}
  Note that $f^*(s)\geq 0$, which guarantees that $y^*(s)+B\leq s$. Thus, if we define $s^* = B+x^* \geq B$, then for $s>s^*$, using \eqref{eq:f-structure}, we have $f^*(s) = \min\{s-s^*, B\}$. Moreover, if $s\le s^*$, then $f^*(s)=0$, which corresponds to the boundary case $y^*(s)+B=s$. This boundary is shared by the second and third cases, with no overshoot.
\end{itemize}
Finally, by summarizing all the above cases, we conclude that there exists an $s^* \geq B$ such that 
\begin{align*}
    f^*(s)=\begin{cases} 
         \min\{s, B\}-B & s \leq s^*\\[10 pt]
         \min\{s-s^*, B\} & s^*<s.
    \end{cases}
\end{align*}    
\end{proof}
\begin{rem}\label{rem:trivial-hold}
    Using similar reasoning, one can show that when $c_{\mathrm{over}} \leq c_{\mathrm{hold}}$, the optimal policy is to schedule as many transactions as possible up to the hard capacity $2B$. That is,
    \[
    f^*(s) = \min\{s-B, B\}.
    \]
\end{rem}

The results in this section imply that, as long as the pool congestion remains below a certain threshold, the optimal policy schedules $B$ units of transactions, or all available transactions if the pool contains fewer than $B$ units. Consequently, the pool congestion increases over time until it reaches a threshold $s^* \geq B$. Once the congestion exceeds $s^*$, the policy schedules more than $B$ units, with the scheduled amount increasing linearly with the congestion level. This reduces the pool congestion, eventually bringing it back below $s^*$, after which the process repeats.

\section{Block Capacity Lower Bounds under Bang--Bang Pricing with Uniform Arrivals}
\label{sec:bangbang}
In this section, we study another special case of the dynamic scheduling problem introduced in Section~\ref{sec:dynamic}, where arrivals are uniform, meaning that at each time step $t$, exactly one transaction of each type arrives.\footnote{Note that although we assume uniform arrivals in this section, most of our derivations extend to the more general case of stationary stochastic arrivals, with the bounds depending on the underlying distributions.} Here, instead of characterizing the structure of the optimal policy, we are interested in obtaining a lower bound on the target block capacity that ensures stability of the system under a simple pricing rule, namely bang--bang pricing policies (Definition~\ref{def:bangbang}).

This analysis studies the same tradeoff between mempool congestion and block-capacity overshoot under a simple pricing policy. A high price makes fewer transactions eligible, allowing the backlog to grow. A low price allows more transactions to be scheduled, but may cause the block size to exceed the target capacity. 

The motivation for considering such bang--bang policies is inspired by our analysis in Section~\ref{sec:threshold}, which suggests that under reasonably homogeneous transaction dynamics, if the scheduler aims to use a simple threshold-based pricing mechanism to control the mempool, it should set a high price to schedule as few transactions as possible while the pool congestion remains below the threshold. As congestion increases and reaches the threshold, the scheduler must lower prices to their minimum value to schedule more transactions and bring congestion back toward the threshold. Moreover, our analysis in this section is particularly useful from a managerial perspective, where the scheduler seeks simple and interpretable pricing policies to control mempool congestion. In particular, it helps determine a minimum target block capacity that ensures even a simple bang--bang policy can schedule all incoming transactions while keeping the average overshoot cost bounded.

\begin{definition}[Bang--Bang Pricing Mechanism]
\label{def:bangbang}
    A \emph{bang--bang pricing mechanism} is a pricing rule under which the posted price at each time step is restricted to the extreme values of the price range. 
    Specifically, the mechanism satisfies $p_t \in \{p_{\min},\, p_{\max}\}
    \ \forall t \ge 1$.
\end{definition}

Before presenting the main result of this section, we introduce the following preliminaries. For simplicity, we work with the long-run average reward rather than the long-run discounted reward. However, the analysis extends naturally to the discounted setting as well. The long-run average reward of a policy $\pi$, starting from the initial state $\mathbf{S}_0$, is defined as follows:
\begin{align*}
    V^\pi(\mathbf{S}_0)
    =\mathbb{E}_\pi
    \left[\lim_{T\to \infty}\frac{1}{T}
    \sum_{t=1}^{T} r(\mathbf{S}_t, a_t) \mid\mathbf{S}_0
    \right]=\sum_{a \in \mathcal{A}} \sum_{\mathbf{S} \in \mathcal{S}} \rho^\pi(\mathbf{S}, a) r(\mathbf{S}, a),
\end{align*}
where
\[
\rho^\pi(\mathbf{S}, a)
=
\lim_{T\to \infty}
\frac{1}{T}
\sum_{t=0}^{T}
\mathbb{P}\{\mathbf{S}^t=\mathbf{S},\, a^t=a\}
\]
denotes the occupancy measure induced by policy \(\pi\), representing the proportion of time that policy \(\pi\) spends at the state-action pair \((\mathbf{S},a)\). Therefore, the  long-run discounted average reward maximization problem can be formulated as
\begin{align}
    \max \quad 
    & \sum_{a \in \mathcal{A}} \sum_{\mathbf{S} \in \mathcal{S}} \rho^\pi(\mathbf{S}, a) r(\mathbf{S}, a) \label{reward-maximization}\\
    \mathrm{s.t.} \quad 
    & \rho^\pi(\mathbf{S}, a) \in \Delta \notag.
\end{align}
where the reward function \(r(\cdot,\cdot)\) is defined in Eq.~\eqref{reward}, and \(\Delta\) denotes the polytope of feasible occupancy measures defined as follows \cite{altman2021constrained}:
\begin{align}
\label{eq:Delta-polytope}
    \Delta
    =
    \left\{
    \rho \in [0,1]^{|\mathcal{S}\times \mathcal{A}|} :
     \sum_{\mathbf{S},a} \rho(\mathbf{S},a) = 1,\ \  
     \sum_{a} \rho(\mathbf{S}, a)
    = \sum_{\mathbf{S}', a'} P(\mathbf{S} \mid \mathbf{S}', a') \rho(\mathbf{S}', a') \ \forall \mathbf{S}\in \mathcal{S}
    \right\}.
\end{align}

It is worth noting that the hard capacity constraint of each block is implicitly incorporated into the feasible polytope $\Delta$ defined in Eq.~\eqref{eq:Delta-polytope}. Specifically, for the prescribed scheduling protocol $\mathcal{F}$, the transition probability in \eqref{eq:prob-tran-matrix} automatically enforces the block capacity constraint, which in turn restricts the set of feasible occupancy measures in \eqref{eq:Delta-polytope}. Thus, no additional capacity constraints are needed in the maximization problem \eqref{reward-maximization}.

\begin{rem}
    The primal optimization problem for the dynamic scheduling problem can be interpreted as finding an optimal stationary policy (the primal variables) that maximizes the long-run discounted average reward, i.e., the value function (the primal objective). Reformulating this problem in terms of occupancy measures (the dual variables) yields the dual program given in \eqref{reward-maximization}. As shown earlier, the price update rule under the Natural Policy Gradient algorithm—which can be thought of as an exponential update rule for the occupancy measures—closely mirrors the dual update rule in the primal-dual analysis of the EIP-1559 mechanism. This correspondence provides a primal-dual extension of the earlier static EIP-1559 analysis to the dynamic setting.     
\end{rem}

\begin{definition}
    We define the \emph{average overshoot cost} associated with a pricing mechanism that schedules $Q_t$ transactions on block $t$ as
    \begin{align}
    \label{eq:overshoot-cost}
        \overline{C}_{\mathrm{over}}
        \coloneqq
        \lim_{T\to\infty}
        \frac{1}{T}\sum_{t=1}^T
        \mathbb{E}\left[c_{\mathrm{over}}(Q_t-B)^+\right].
    \end{align}   
\end{definition}

\begin{theorem}
\label{theorem:bangbang}
    Let $Q = \sum_{i=1}^m q_i$ and  $|\mathcal{V}|=n$. Suppose that the target block capacity $B$ is chosen such that 
    \begin{equation}
    \label{eq:B-lower-bound-theorem}
        \frac{Q}{4}\left(3+\sqrt{8n^2-16n+9}\right) \leq B,
    \end{equation}
    and that the interval
    \begin{equation}\label{k-range-theorem}
        \left[
        \frac{Q(n-1)}{B-Q},
        \,
        \frac{2B-Q}{Q(n-1)}
        \right]
    \end{equation}
    contains a positive integer.\footnote{As shown in Remark \ref{rem:capacity}, it is always possible to satisfy both of these conditions simultaneously.} Then there exists a bang--bang pricing mechanism that can schedule all incoming transactions and the average overshoot cost would be upper bounded by
    \[
    \overline{C}_{\mathrm{over}}
    \leq \frac{BQ(n-1)}{2B -Q}.
    \]
\end{theorem}

\begin{proof}
   Please see Appendix~\ref{proof-bangbang} for the proof.
\end{proof}

\begin{rem}\label{rem:capacity}
    Note that a positive integer in the interval defined by \eqref{k-range-theorem} can always be found by increasing the block capacity $B$. In particular, if $B$ is increased by an amount $x$, the upper bound in \eqref{k-range-theorem} increases by $\frac{2x}{Q(n-1)}$, while the lower bound decreases by $\frac{x Q (n-1)}{(B - Q)(B + x - Q)}.$
    As a result, the feasible interval expands as $B$ increases, making it always possible to find a positive integer within this range by choosing $B$ sufficiently large.
\end{rem}

\section{Numerical Results}\label{sec:numerical-results}
In this section, we first present numerical results to evaluate the performance of Algorithm~\ref{alg:NPG}. We consider the general dynamic transaction scheduling problem and study the algorithm under both deterministic and stochastic transaction arrivals while varying the overshoot penalty $c_{\mathrm{over}}$. We then consider the homogeneous transaction setting and numerically validate the threshold structure of the optimal policy predicted by the theoretical analysis in Section~\ref{sec:threshold}. Additional numerical experiments, including an evaluation of the bang--bang pricing mechanism proposed in Section~\ref{sec:bangbang}, are presented in Appendix~\ref{appendix:numerical}.

\subsection{Performance of Algorithm~\ref{alg:NPG}}
Throughout this subsection, we run Algorithm~\ref{alg:NPG} for a total of $10{,}000$ iterations, where each iteration consists of an episode of length $H=400$. The discount factor is fixed at $\gamma = 0.95$, and the learning rate is tuned separately for each experiment while satisfying the condition in Eq.~\eqref{eq:lr-condition}.

Unless otherwise specified, the holding cost is fixed at $c_{\mathrm{hold}} = 1$, and the overshoot penalty is varied over
\[
c_{\mathrm{over}} \in \{1, 5, 10, 30, 50, 100, 200, 500\}.
\]
We consider the following three settings, each corresponding to a different choice of transaction sizes and per-unit values:
\begin{itemize}
    \item \emph{Setting~1}: $\mathcal{Q}=\{2, 4, 5, 7\}$ and $\mathcal{V}=\{2.0, 3.0, 4.0, 5.0, 9.0\}$.
    \item \emph{Setting~2}: $\mathcal{Q}=\{1, 2, 8, 15\}$ and $\mathcal{V}=\{1.0, 2.0, 3.0, 10.0, 25.0\}$.
    \item \emph{Setting~3}: $\mathcal{Q}=\{2, 3, 4, 5\}$ and $\mathcal{V}=\{1.0, 1.5, 2.0, 2.5, 150.0\}$.
\end{itemize}

For the experiments below, we assume a uniform deterministic arrival process in which exactly one transaction of each type arrives at every time step. For each setting, we plot the average scheduled transaction volume as a function of the overshoot penalty for three different target block capacities. The range of feasible target block capacities is chosen according to Remark~\ref{B-Q-range}. Figure~\ref{fig:deterministic-arrivals} presents the results for all three settings. From Fig.~\ref{fig:deterministic-arrivals}, we observe that as the overshoot penalty increases, the average scheduled transaction volume moves toward the target block capacity, matching the behavior exhibited by the EIP-1559 mechanism \cite{babaioff2024optimality}. Furthermore, for each setting, increasing the target block capacity enables the algorithm to more closely match the target block capacity when the overshoot penalty is sufficiently large.
 
\begin{figure}[htb]
    \centering
    \begin{subfigure}{0.32\textwidth}
        \centering
        \includegraphics[width=\linewidth]{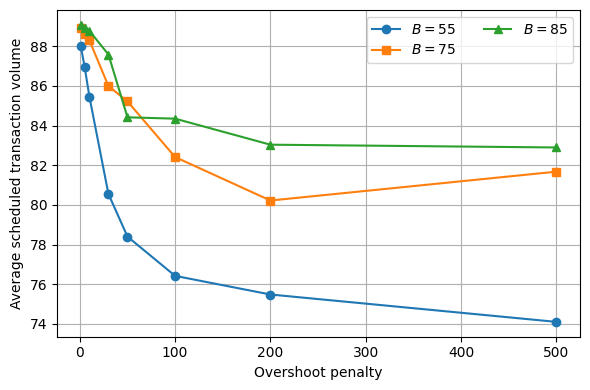}
        \caption{Setting~1}
        \label{fig:setting1-det}
    \end{subfigure}
    \hfill
    \begin{subfigure}{0.32\textwidth}
        \centering
        \includegraphics[width=\linewidth]{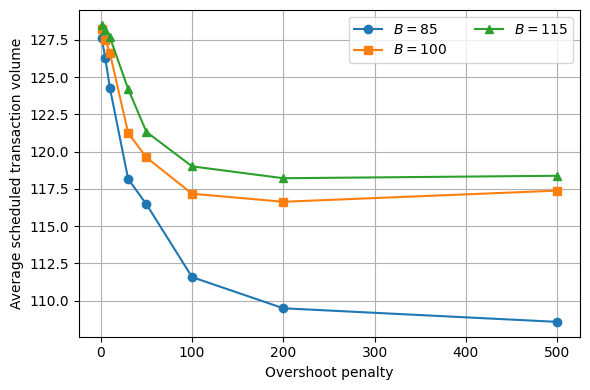}
        \caption{Setting~2}
        \label{fig:setting2-det}
    \end{subfigure}
    \hfill
    \begin{subfigure}{0.32\textwidth}
        \centering
        \includegraphics[width=\linewidth]{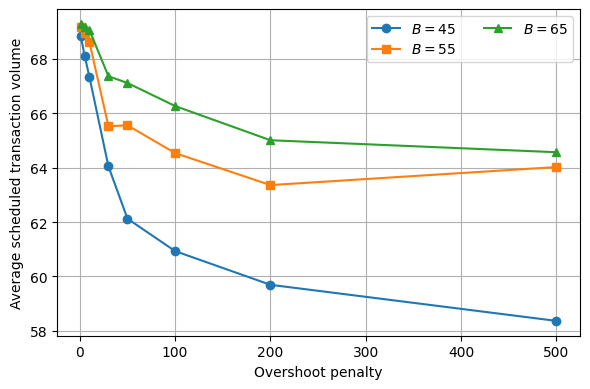}
        \caption{Setting~3}
        \label{fig:setting3-det}
    \end{subfigure}
    \caption{Average scheduled transaction volume as a function of the overshoot penalty for all three settings.}
    \Description{Average scheduled transaction volume as a function of the overshoot penalty for all three settings.}
    \label{fig:deterministic-arrivals}
\end{figure}

We now extend our simulation results to stochastic arrivals, assuming that transactions enter the pool according to a stationary random process. Let $A[i,j] \sim \mathrm{Poisson}(\mu_i)$ denote the number of arriving transactions of size $q_i$ and value $v_j$ in a single time step. The total arrival size in one time step is then given by
\begin{equation}
Q_{\mathrm{arr}}
\coloneqq
\sum_{i=1}^m \sum_{j=1}^n q_i\  A[i,j].
\end{equation}
Letting $\mu_i = \frac{2B\lambda}{m n q_i}$, where $\lambda\in(0,1)$, the expected total arrival size is
\begin{align}\nonumber
\mathbb{E}\left[ Q_{\mathrm{arr}} \right] = \sum_{i=1}^m \sum_{j=1}^n q_i \ \mathbb{E}\left[ A[i,j] \right] = \sum_{i=1}^m \sum_{j=1}^n q_i \mu_i = \sum_{i=1}^m n q_i \left( \frac{2B}{m n q_i}\lambda \right) = \sum_{i=1}^m \frac{2B}{m}\lambda = 2B\lambda.
\end{align}

Therefore, to avoid the trivial case, we choose $\lambda \in (0.5, 1)$, which ensures that the expected total arrival size exceeds the target block size $B$ while remaining below the hard block capacity $2B$.

We now present the results for Setting~1 under stochastic arrivals and defer the results for Settings~2 and~3 to Appendix~\ref{appendix:numerical}. Unlike the deterministic arrival case, here we can vary the target block size $B$ over a wide range without being constrained by a fixed number of arrivals. Figure~\ref{fig:Stochastic-Setting1} shows the average scheduled transaction volume as a function of the overshoot penalty for three different arrival rates $\lambda$ and three different target block capacities. As shown in Fig.~\ref{fig:Stochastic-Setting1}, the average scheduled transaction volume moves toward the target block capacity as the overshoot penalty increases, mirroring the behavior observed in the deterministic arrival setting. However, as the arrival rate $\lambda$ increases, transaction counts exceed the threshold $L$ introduced in Remark~\ref{rem:schedul-F} more frequently, requiring the algorithm to prioritize scheduling the excess transactions. As a result, the algorithm has fewer opportunities to strategically control the scheduled transaction volume, and its behavior becomes less sensitive to the choice of overshoot penalty.
\begin{figure}[htb]
    \centering
    \begin{subfigure}{0.32\textwidth}
        \centering
        \includegraphics[width=\linewidth]{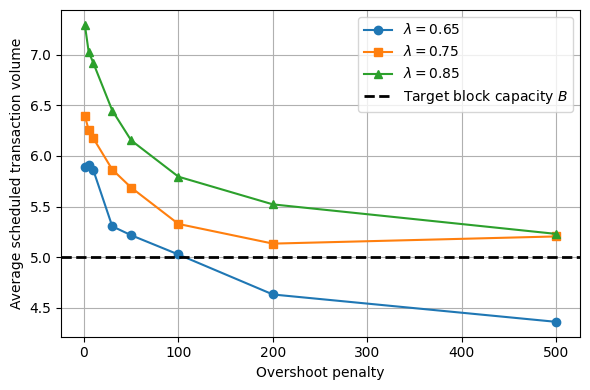}
        \caption{$B=5$}
        \label{fig:S1,B=5}
    \end{subfigure}
    \hfill
    \begin{subfigure}{0.32\textwidth}
        \centering
        \includegraphics[width=\linewidth]{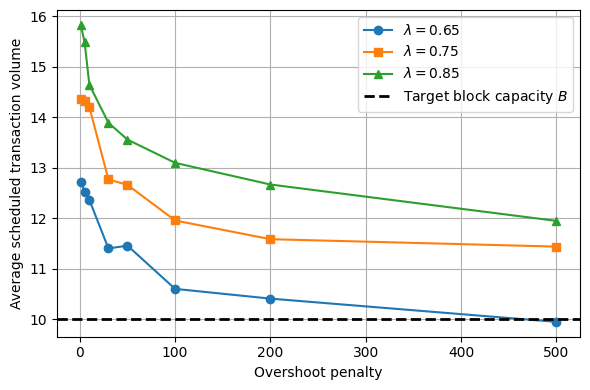}
        \caption{$B=10$}
        \label{fig:S1,B=10}
    \end{subfigure}
    \hfill
    \begin{subfigure}{0.32\textwidth}
        \centering
        \includegraphics[width=\linewidth]{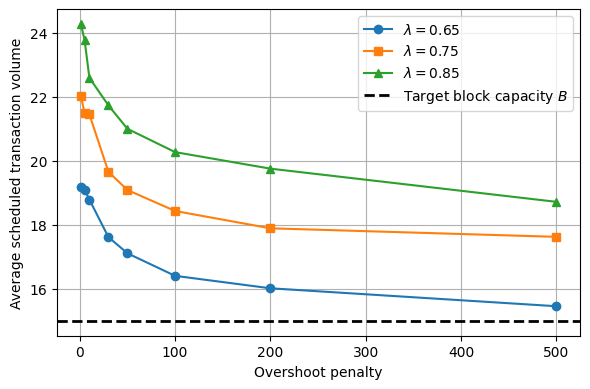}
        \caption{$B=15$}
        \label{fig:S1,B=15}
    \end{subfigure}
    \caption{Average scheduled transaction volume as a function of the overshoot penalty for Setting~1 under stochastic arrivals with three different target block capacities.}
    \Description{Average scheduled transaction volume as a function of the overshoot penalty for Setting~1 under stochastic arrivals with three different target block capacities.}
    \label{fig:Stochastic-Setting1}
\end{figure}

\subsection{Homogeneous Transactions}
In this subsection, we numerically validate the threshold structure established in Theorem~\ref{threshold-policy}. We set the target block capacity to $B=10$, the maximum state value to $L=100$, and the discount factor to $\gamma=0.95$. The optimal policy is computed using value iteration with a maximum of $T=500$ iterations. Throughout the experiments, the total size of transactions arriving at each time step is modeled as an i.i.d.\ random variable uniformly distributed over the interval $[0,2B]$. Furthermore, consistent with the assumptions in Section~\ref{sec:threshold}, we assume that $c_{\mathrm{over}} > c_{\mathrm{hold}}$. Figure~\ref{fig:threshold-cover-range} illustrates the resulting optimal policy for $c_{\mathrm{hold}}=1$ and several values of the overshoot penalty $c_{\mathrm{over}}$.

\begin{figure}[htb]
    \centering
    \begin{subfigure}{0.32\textwidth}
        \centering
        \includegraphics[width=\linewidth]{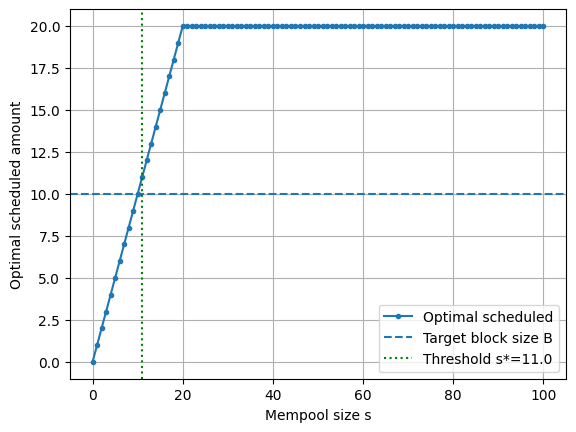}
        \caption{$c_\mathrm{over}=2$}
        \label{fig:thresold-c2}
    \end{subfigure}
    \hfill
    \begin{subfigure}{0.32\textwidth}
        \centering
        \includegraphics[width=\linewidth]{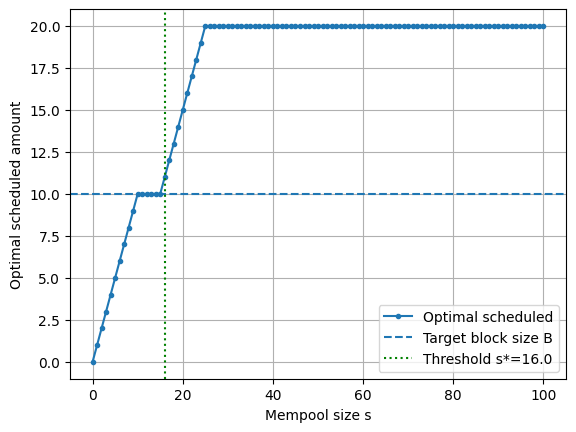}
        \caption{$c_\mathrm{over}=3.5$}
        \label{fig:thresold-c3.5}
    \end{subfigure}
    \hfill
    \begin{subfigure}{0.32\textwidth}
        \centering
        \includegraphics[width=\linewidth]{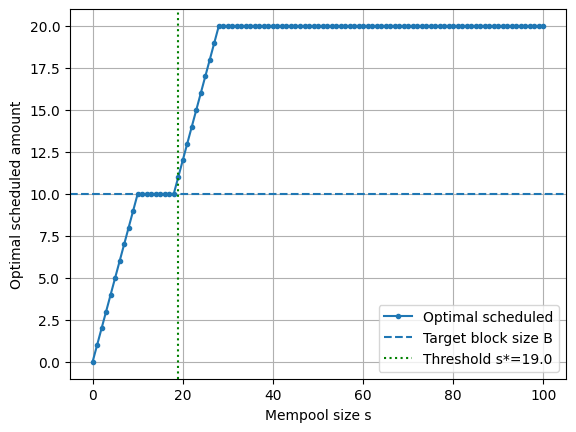}
        \caption{$c_\mathrm{over}=5$}
        \label{fig:thresold-c5}
    \end{subfigure}
    \caption{Optimal scheduled transaction volume as a function of the mempool size for three values of $c_{\mathrm{over}}$, with $c_{\mathrm{hold}}=1$.}
    \Description{}
    \label{fig:threshold-cover-range}
\end{figure}

As illustrated in Fig.~\ref{fig:threshold-cover-range}, when the mempool size is below the threshold $s^*$, the optimal scheduling decision is $\min\{s,B\}$. Once the mempool size exceeds $s^*$, the optimal number of scheduled transactions increases with the level of mempool congestion until it reaches the maximum block capacity of $2B$. We also observe that the threshold $s^*$ increases as the overshoot penalty $c_{\mathrm{over}}$ becomes larger.

\section{Conclusions}
Ethereum’s EIP-1559 mechanism regulates block utilization through a simple price update rule. However, existing analyses treat this pricing problem in a static setting and do not explicitly model the evolving composition of the mempool. In this paper, we study a dynamic transaction scheduling problem with heterogeneous transaction sizes and per-unit values, where users are \emph{patient} and unscheduled transactions remain in the mempool and can be scheduled in future blocks.

We first provide a primal--dual interpretation of the static EIP-1559 mechanism, showing that the block price naturally arises as a dual variable in a linear program that maximizes the total value of scheduled transactions subject to block capacity constraints. We then leverage this view to analyze the mechanism via a primal--dual competitive-ratio argument. Next, we extend this framework to a dynamic setting by formulating the scheduling problem as a discounted MDP whose state captures the mempool configuration and whose actions correspond to prices. We further incorporate holding costs and overshoot penalties into a long-run discounted objective and apply an episodic Natural Policy Gradient (NPG) algorithm to learn an optimal dynamic pricing policy. Our experiments in both deterministic and stochastic arrival regimes confirm that as the overshoot penalty increases, the average scheduled transaction volume converges to the target block capacity $B$, and the resulting NPG updates closely resemble the EIP-1559 exponential price update rule observed in the static setting. Finally, we examine two special cases of the MDP formulation: one with homogeneous transactions and one with uniform arrivals. In the homogeneous setting, where the protocol directly controls the scheduled transaction volume, we show that the optimal policy has a threshold structure. Building on this insight and addressing block overflow, we then derive a theoretical lower bound on the block capacity $B$ required to ensure system stability under a bang--bang pricing mechanism with uniform arrivals.

Our work opens an avenue to several intriguing research questions. One direction is to extend our results to a partially observable setting. Another is to characterize the structure of optimal scheduling policies for a broader class of arrival processes beyond homogeneous or uniform arrivals. Finally, extending our results to scenarios where the scheduler is further constrained by the environment is an important direction. For instance, one may ask how to design optimal policies when transactions are strategic and the scheduler must design a mechanism to elicit truthful information about users’ transaction sizes and values in a dynamic manner \cite{leon2025online}. Another extension is to consider heterogeneous blocks with varying features, which impose additional constraints on the types of transactions that can be scheduled. In such constrained dynamic environments, we believe that our primal--dual analysis based on occupancy measures sets the stage for incorporating additional constraints and future developments.

\bibliographystyle{ACM-Reference-Format}
\bibliography{sample-base}

\appendix
\section*{Appendix}

\section*{Notations}
The main notations used in this paper are listed in Table~\ref{tab:notations}.
\begin{table}[htb]
\caption{List of notations.}
\centering
\small  
\setlength{\tabcolsep}{4pt}
\scalebox{0.9}{
\begin{tabular}{|c|c|}
\hline
\textbf{Notation} & \textbf{Definitions} \\
\hline
$\mathcal{Q}=\{q_1, \cdots, q_m\}$ & set of transaction sizes\\
$\mathcal{V}=\{v_1, \cdots, v_n\}$ & set of transaction per unit values\\
$m$ & number of different transaction sizes \\
$n$ & number of different per unit values \\
$B$ & target block size \\
$c \cdot B$ & maximum block size\\
$r_j$ & arrival time of transaction $j$ \\
$x_{jt}$ & fraction of transaction $j$ scheduled on block $t$ \\
$\mathcal{C}$ & set of all transactions \\
$a_t$ & action at time $t$ \\
$\mathbf{S}_t$ & state at time $t$ \\
$\mathbf{S}_0$ & initial state \\
$\mathcal{A}$ & action space \\
$\mathcal{S}$ & state space \\
$\mathcal{F}$ & scheduling protocol \\ 
$\rho^\pi(\cdot, \cdot)$ & occupancy measure induced by policy $\pi$ \\
\hline
\end{tabular}}
\label{tab:notations}
\end{table}

\section{An Alternative Primal-Dual Formulation for EIP-1559}
\label{app}

In this appendix, we provide a more detailed primal--dual formulation of EIP-1559, which yields an interpretable price update rule in terms of the corresponding dual variables.

Let $Q_t \coloneqq \sum_{j \in \mathcal{C}} q_j x_{jt}$ denote the total scheduled size in block $t$. We aim to incorporate a one-sided penalty $(Q_t - B)_+$ into our main objective, which is the total value of scheduled transactions:
\[
\sum_t \left[ \sum_{j \in \mathcal{C}} v_j x_{jt} - c_{\mathrm{over}} (Q_t - B)_+ \right],
\]
where $c_{\mathrm{over}}$ is the penalty coefficient for exceeding the target block size $B$.
 
\[
\mathrm{SW}:= \sum_{j\in\mathcal C}\sum_{t=1}^T q_j v_j\,x_{jt}.
\]
To model the one-sided penalty $(Q_t-B)_+$, we introduce auxiliary variables $y_t\ge 0$ such that
$y_t \ge Q_t - B$ and maximize welfare minus $\sum_t y_t$:
\begin{align}
\max \quad 
& \sum_{j\in\mathcal C}\sum_{t=1}^T q_j v_j\,x_{jt}\;-\;\sum_{t=1}^T y_t \label{eq:P-onesided}\\
\text{s.t.}\quad
& \sum_{t=r_j}^T x_{jt} \le 1 \qquad\qquad\qquad\ \ \forall j\in\mathcal C, \notag\\
& \sum_{j\in\mathcal C} q_j x_{jt} \le 2B \qquad\qquad\qquad\ \ \forall t\in\{1,\dots,T\}, \notag\\
& \sum_{j\in\mathcal C} q_j x_{jt} - y_t \le B \qquad\qquad\ \ \forall t\in\{1,\dots,T\}, \notag\\
& y_t \ge 0 \qquad\qquad\qquad\qquad\qquad\ \ \forall t\in\{1,\dots,T\}, \notag\\
& x_{jt}\ge 0 \qquad\qquad\qquad\qquad\qquad\ \ \forall j\in\mathcal C,\ \forall t\ge r_j. \notag
\end{align}

Then the dual can be written as
\begin{align}
\min \quad 
& \sum_{j\in\mathcal C}\alpha_j + \sum_{t=1}^T \beta_t + \sum_{t=1}^T \gamma_t \label{eq:D-onesided}\\
\text{s.t.}\quad
& \alpha_j \;+\; q_j\left(\frac{\beta_t}{2B}+\frac{\gamma_t}{B}\right)
\ \ge\ q_j v_j
\qquad \forall j\in\mathcal C,\ \forall t\ge r_j, \notag\\
& \alpha_j\ge 0 \quad \forall j,\qquad
\beta_t \ge 0 \quad \forall t,\qquad
0\le \gamma_t \le B \quad \forall t. \notag
\end{align}

We now provide an interpretation of the dual variables. Let $\beta_t$ denote the shadow price associated with the hard capacity constraint of block $t$, i.e., the marginal cost of utilizing the full block capacity $2B$. Accordingly, scheduling $q_j$ units of transaction $j$ in block $t$ incurs a capacity cost of
\[
\frac{\beta_t}{2B} \, q_j.
\]

The dual variable $\gamma_t$ represents the penalty associated with exceeding the target block size $B$ in block $t$. Specifically, $\gamma_t$ is activated when the total scheduled transaction size on block $t$, i.e., $Q_t$, exceeds $B$, and it captures the cost of allocating up to an additional block-sized amount of capacity beyond the target. Equivalently, when the block size exceeds $B$, scheduling $q_j$ additional units of transaction $j$ incurs an extra charge of
\[
\frac{\gamma_t}{B} \, q_j.
\]
We can also interpret this as follows: if $Q_t > B$, then the share of the overshoot penalty allocated to transaction $j$ is proportional to its size, i.e., $q_j$ multiplied by the marginal cost $\frac{\gamma_t}{B}$.

Together, the effective per-unit price faced by a transaction scheduled on block $t$ is
\[
p_t = \frac{\beta_t}{2B} + \frac{\gamma_t}{B},
\]
which reflects both the scarcity of total block capacity and the congestion cost of exceeding the target block size.

Transaction $j$ has total value $v_j q_j$. Thus, if it is scheduled on block $t$, its utility is
\[
u_{j,t} = q_j v_j - \frac{\beta_t}{2B} q_j - \frac{\gamma_t}{B} q_j.
\]

Define
\[
\alpha_j \coloneqq \max_{t \ge r_j} \left\{ q_j v_j - \frac{\beta_t}{2B} q_j - \frac{\gamma_t}{B} q_j \right\},
\]
which represents the maximum utility that transaction $j$ can achieve. Since transactions never accept negative utility, transaction $j$ is scheduled only on blocks whose effective price does not exceed its value, i.e.,
\[
p_t = \frac{\beta_t}{2B} + \frac{\gamma_t}{B} \le v_j.
\]

Therefore, the problem of choosing the allocation variables $x_{jt}$ reduces to the problem of setting an appropriate price for each block $t$. As a result, only transactions with $v_j \ge p_t$ can be scheduled in block $t$.

\section{An Alternative Competitive Ratio for EIP-1559}\label{appx:alternative-competitive-ratio}
In this appendix, we derive an alternative competitive ratio bound for EIP-1559, stated in Theorem~\ref{theorem:competitive-ratio}, without explicitly using the price update rule. Instead, the analysis relies only on the fact that EIP-1559 schedules at least $q_{\mathrm{min}}\leq 2B$ units in every block. Although this yields a tighter competitive ratio, it does not capture the effect of the EIP-1559 pricing dynamics.
We adopt the notation from the proof of Theorem~\ref{theorem:competitive-ratio}. In particular, define $
\beta_t =2Bp_t\frac{v_{\max}}{p_{\min}},$ and $\alpha_j\geq0$, so that $(\alpha,\beta)$ is a feasible solution of the dual program. Hence, by weak duality,
\[
V^{\mathrm{OPT}}
\le
\sum_{t=1}^T\beta_t+\sum_{j\in\mathcal C}\alpha_j.
\]
For each transaction $j \in \mathcal{C}_{\mathrm{sched}}$, let $t_j$ denote the block in which the online algorithm schedules transaction $j$. Also, define
    \[
    \alpha_j =
    \begin{cases}
    q_jv_j - p_{t_j}q_j, & j \in \mathcal{C}_{\mathrm{sched}},\\
    0, & j \in \mathcal{C}\setminus\mathcal{C}_{\mathrm{sched}}.
    \end{cases}
    \]
    Note that $\alpha_j \ge 0$ for all $j \in \mathcal{C}$. 
    
    We now derive a lower bound on $V^\text{OA}$ that simultaneously bounds both dual variables:
    \begingroup
    \allowdisplaybreaks[1]
    \begin{align*}
    V^{\text{OA}}
    &= \sum_{j\in\mathcal{C}_{\mathrm{sched}}} v_j q_j\\
    &= \sum_{j\in\mathcal{C}_{\mathrm{sched}}} \left(v_j q_j - p_{t_j}q_j + p_{t_j}q_j \right) \\
    &= \sum_{j\in\mathcal{C}_{\mathrm{sched}}} \left(\alpha_j + p_{t_j}q_j\right) \\&=\sum_{j\in\mathcal{C}_{\mathrm{sched}}} \alpha_j + \sum_{t=1}^T p_t \sum_{j\in\mathcal{C}_{\mathrm{sched}}: t_j=t} q_j
    \\
    &\ge \sum_{j\in\mathcal{C}} \alpha_j + q_{\mathrm{min}} \sum_{t=1}^T p_t \notag\\
    &= \sum_{j\in\mathcal{C}} \alpha_j
       +  \frac{q_{\mathrm{min}}p_{\mathrm{min}}}{2Bv_{\mathrm{max}}}\sum_{t=1}^T \beta_t \\
    &\ge \min\{1, \frac{q_{\mathrm{min}}p_{\mathrm{min}}}{2Bv_{\mathrm{max}}}\} \left(\sum_{j\in\mathcal{C}}\alpha_j + \sum_{t=1}^T \beta_t \right).
    \end{align*}
    \endgroup
    Therefore the competitive ratio is 
    \[
    \gamma = \frac{2Bv_{\mathrm{max}}}{q_{\mathrm{min}}p_{\mathrm{min}}}.
    \]

\section{Auxiliary Lemmas and Omitted Proofs}

\subsection{Auxiliary Lemmas}
\begin{lemma}
\label{convexity_min_jointly-convex}
    Let $h:\mathbb{R}^2 \to \mathbb{R}$ be a jointly convex function, and define $g(x) := \inf_{f\in \mathbb{R}} h(x,f)$. Then, $g$ is a convex function on $\mathbb{R}$. 
\end{lemma}
\begin{proof}
Let $g(x_1) = \inf_{f\in \mathbb{R}} h(x_1,f)$ and $g(x_2) = \inf_{f\in \mathbb{R}} h(x_2,f)$. For every $\epsilon>0$, there exists $f_1, f_2\in \mathbb{R}$ such that $g(x_1) + \epsilon \geq h(x_1, f_1)$ and $g(x_2) + \epsilon \geq h(x_2, f_2)$.
    For $\lambda \in [0,1]$, we have
    \begin{align*}
        \lambda g(x_1) + (1-\lambda)g(x_2) &= 
        \lambda (g(x_1) + \epsilon) + (1-\lambda)(g(x_2) + \epsilon) - \epsilon \\& \geq
        \lambda h(x_1, f_1) + (1-\lambda) h(x_2, f_2) - \epsilon
        \\& \geq 
        h(\lambda x_1 + (1-\lambda)x_2, \lambda f_1 + (1-\lambda)f_2) - \epsilon
        \\& \geq g(\lambda x_1 + (1-\lambda)x_2) - \epsilon.
    \end{align*}
    Since this relation holds for every $\epsilon>0$, we conclude that 
    \begin{align*}
        \lambda g(x_1) + (1-\lambda)g(x_2) \geq g(\lambda x_1 + (1-\lambda)x_2),
    \end{align*}
    and therefore $g(\cdot)$ is convex on $\mathbb{R}$.
\end{proof}

\begin{lemma} 
\label{convexity_over_convex_set}
Let $Q:\mathbb{R}^2 \to \mathbb{R}$ be jointly convex and continuous, and let $\mathcal{H} \subseteq \mathbb{R}^2$ be a nonempty convex set. For each $x \in \mathbb{R}$, define $\mathcal{H}_x := \{f \in \mathbb{R} : (x,f) \in \mathcal{H}\}$, and suppose that $\mathcal{H}_x$ is compact. Define $g(x) := \min_{f \in \mathcal{H}_x} Q(x,f)$. Then $g$ is convex on the set $\{x \in \mathbb{R} : \exists\, f \in \mathbb{R} \text{ such that } (x,f) \in \mathcal{H}\}$. 
\end{lemma}    
\begin{proof}
Since $Q$ is jointly convex, $Q(x,\cdot)$ must be convex for each fixed $x$. Also, since $\mathcal{H}_x$ is assumed to be compact, $\min_{f\in \mathcal{H}_x} Q(x,f)$ is well-defined. Define $h: \mathbb{R}^2 \to \mathbb{R}$ as 
    \begin{align*}
        h(x,f)=\begin{cases} 
             0 & (x,f)\in\mathcal{H}, \\[10pt]
            \infty & \text{else},
        \end{cases}
    \end{align*}
    where we note that $h: \mathbb{R}^2 \to \mathbb{R}$ is jointly convex. This is because for $\lambda\in[0,1]$, we have
    \begin{align*}
        \lambda h(x_1, f_1) + (1-\lambda)h(x_2, f_2) \geq h(\lambda x_1 + (1-\lambda)x_2, \lambda f_1+(1-\lambda)f_2),
    \end{align*}
    which holds if $(x_1, f_1), (x_2, f_2)\in \mathcal{H}$ because of convexity of $\mathcal{H}$. On the other hand, if one of the points is not in $\mathcal{H}$, for example $(x_1, f_1)\notin \mathcal{H}$ or $(x_2, f_2)\notin \mathcal{H}$, then the left-hand side would be infinity and the inequality would hold trivially.
    
    We can therefore rewrite $g(\cdot)$ as $g(x) = \min_{f \in \mathbb{R}} \{Q(x,f)+h(x,f)\}$, which is well defined since we assume that $\mathcal{H}$ is nonempty. Now since both $h$ and $Q$ are jointly convex, their sum is also jointly convex and according to Lemma~\ref{convexity_min_jointly-convex}, $g$ is convex on its domain.
\end{proof}

\subsection{Proof of Lemma~\ref{J-k-convex}}
\label{proof-J-k-convex}

We first use induction to show that $J_k$ is convex for all $k$. The base case holds since $J_0(x) = c_{\mathrm{term}}x$, which is convex. Now suppose that $J_k(\cdot)$ is convex. Let
\begin{equation}
    J_{k+1}(x) = \min_{f\in\mathcal{H}_x}\{Q_k(x,f)\}, 
\end{equation}
where $Q_k(x,f) := c_{\mathrm{hold}} x + c_{\mathrm{over}} (f)^{+}+\gamma \mathbb{E}[J_{k}(x+A-f-B)]$, and $\mathcal{H}_x = \{f\in \mathbb{R}: (x,f)\in\mathcal{H}\}$ with $\mathcal{H} = \{(x,f)\in [0, \infty)\times [-B, \infty):  f\leq \min\{x, 2B\}-B\}$. We have $(0,-B)\in \mathcal{H}$, hence $\mathcal{H}$ is nonempty. The feasible region $\mathcal{H}$ is clearly closed and convex, but not necessarily bounded. However, for each fixed $x$, the set $\mathcal{H}_x$ is both closed and bounded (see Fig.~\ref{fig:areaH} for $B=5$).
\begin{figure}[htbp]
    \centering
\includegraphics[width=0.45\textwidth]{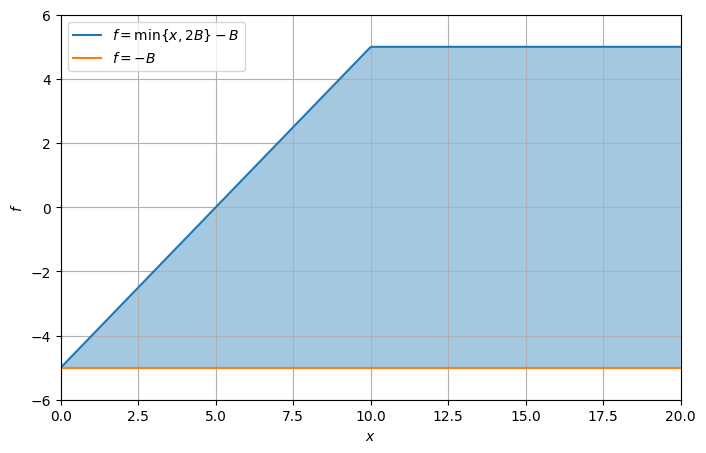}
    \caption{Illustration of the feasible region $\mathcal{H}$ for $B=5$.}
    \Description{Illustration of the feasible region $\mathcal{H}$ for $B=5$.}
    \label{fig:areaH}
\end{figure}

Next, we show that $Q_k(\cdot,\cdot)$ is jointly convex. For any $\lambda \in [0,1]$, it suffices to show that
\begin{equation}\nonumber
    \lambda Q_k(x_1, f_1) + (1-\lambda)Q_k(x_2, f_2) \geq Q_k(\lambda x_1 + (1-\lambda)x_2, \lambda f_1 + (1-\lambda)f_2).
\end{equation}
We can write
\begin{align*}
    \lambda Q_k(x_1, f_1) + (1-\lambda)Q_k(x_2, f_2) &= \lambda \left(c_{\mathrm{hold}} x_1 + c_{\mathrm{over}} (f_1)^{+} + \gamma\mathbb{E}[J_{k}(x_1+A-f_1-B)]\right) \\&\qquad+ (1-\lambda) \left(c_{\mathrm{hold}} x_2 + c_{\mathrm{over}} (f_2)^{+} + \gamma\mathbb{E}[J_{k}(x_2+A-f_2-B)]\right) \\&= c_{\mathrm{hold}} \left(\lambda x_1 + (1-\lambda)x_2\right) + c_{\mathrm{over}} \left(\lambda (f_1)^{+} + (1-\lambda) (f_2)^{+}\right) \\&\qquad+ \gamma \mathbb{E}[\lambda J_{k}(x_1+A-f_1-B) + (1-\lambda) J_{k}(x_2+A-f_2-B)] \\& \geq 
    c_{\mathrm{hold}} \left(\lambda x_1 + (1-\lambda)x_2\right) + c_{\mathrm{over}} (\lambda f_1 + (1-\lambda)f_2)^{+} \\&\qquad+ \gamma \mathbb{E} [J_{k}(\lambda x_1 + (1-\lambda)x_2-\lambda f_1 - (1-\lambda)f_2 + A - B)].
\end{align*}
In the above relations, we have used the fact that $(\cdot)^+$ is convex, and that $J_k(x + A - f - B)$ is the composition of a convex function (by the induction hypothesis) with an affine function of $(x,f)$, and is therefore convex. Moreover, taking expectations preserves convexity.

Therefore, we have shown that $Q_k(\cdot,\cdot)$ is jointly convex, which, in view of Lemma~\ref{convexity_over_convex_set}, implies that $J_{k+1}$ is convex, completing the induction. Now, since we have shown that $J_k(\cdot)$ is convex for all $k = 1,2,\ldots$, and the limit $\lim_{k \to \infty} J_k$ exists, the limit function $J^* := \lim_{k \to \infty} J_k$ is also convex.

Finally, using the convexity of $J^*$ and a similar argument as above, one can show that $Q^*(s,f)$ is jointly convex, where
\[
J^*(s) = \min_{-B \leq f \leq \min\{s,2B\} - B} Q^*(s,f).
\]
Consequently, for any fixed $s$, the function $Q^*(s,\cdot)$ is convex in $f$.

\subsection{Proof of Lemma~\ref{J-non-decreasing}}
\label{proof-J-non-decreasing}
 We first show that $J_k(\cdot)$, as defined in Eq.~\eqref{J-finite-horizon}, is non-decreasing for all $k$ using induction. The statement holds for $J_0(x) = c_{\mathrm{term}} x$. Now suppose that $J_k(\cdot)$ is non-decreasing. Let $x_1 \leq x_2$. We have
    \begin{align*}
        J_{k+1}(x_1)-J_{k+1}(x_2) &= \min_{f\in \mathcal{H}_{x_1}}\left[c_{\mathrm{hold}} x_1 + c_{\mathrm{over}} (f)^{+} + \gamma \mathbb{E}[J_k(x_1 + A - f - B)]\right] \\&- \min_{f\in \mathcal{H}_{x_2}}\left[c_{\mathrm{hold}} x_2 + c_{\mathrm{over}} (f)^{+} + \gamma \mathbb{E}[J_k(x_2 + A - f - B)]\right] 
        \\&:= \min_{f\in \mathcal{H}_{x_1}} \left[Q_k(x_1,f)\right] - \min_{f\in \mathcal{H}_{x_2}} \left[Q_k(x_2,f)\right].
    \end{align*}
    Note that since we have already shown that $J_k$ is convex in Lemma~\ref{J-k-convex}, the quantity $\min_{f \in \mathcal{H}_x} Q_k(x,f)$ is well-defined. Also, by the definition of $\mathcal{H}_x$, we know that $\mathcal{H}_{x_1} \subseteq \mathcal{H}_{x_2}$. We consider two cases:
    \begin{itemize}
        \item $\arg\min_{f \in \mathcal{H}_{x_2}} Q_k(x_2,f) \in \mathcal{H}_{x_1}$: Then, for every realization of $A$, and any $f \in \mathcal{H}_{x_1}$, we have
        \[x_1 + A - f - B \leq x_2 + A - f - B,\]
        and since $J_k(\cdot)$ is non-decreasing,
        \[J_k(x_1 + A - f - B) \leq J_k(x_2 + A - f - B).\]
        Taking expectation preserves the inequality, hence
        \[\mathbb{E}[J_k(x_1 + A - f - B)] \leq \mathbb{E}[J_k(x_2 + A - f - B)].\]
        And since $x_1 \leq x_2$, we have $Q_k(x_1, f)\leq Q_k(x_2,f).$ Taking minimum over $f \in \mathcal{H}_{x_1}$ we obtain
        \[
        J_{k+1}(x_1) = \min_{f \in \mathcal{H}_{x_1}} Q_k(x_1,f)\leq \min_{f \in \mathcal{H}_{x_1}} Q_k(x_2,f) = \min_{f \in \mathcal{H}_{x_2}} Q_k(x_2,f)=J_{k+1}(x_2).
        \]
        \item $\arg\min_{f\in\mathcal{H}_{x_2}}Q_k(x_2,f)\in\mathcal{H}_{x_2}-\mathcal{H}_{x_1}$: In this case, let us define $f_2^*=\arg\min_{f\in\mathcal{H}_{x_2}}Q_k(x_2,f)$ and $f_1^* = \arg\min_{f\in\mathcal{H}_{x_1}}Q_k(x_1,f)$. According to the definition of $\mathcal{H}_x$, it is not hard to see that \[f_1^*\leq \min\{x_1,2B\}-B<f_2^*\leq \min\{x_2,2B\}-B.\]
        Instead of directly showing $Q_k(x_1, f_1^*)\leq Q_k(x_2, f_2^*)$, we will show that $Q_k(x_1, f_1^*)\leq Q_k(x_1, \hat{f})$, and $Q_k(x_1, \hat{f})\leq Q_k(x_2, f_2^*)$, where $\hat{f}=\min\{x_1, 2B\}-B\in\mathcal{H}_{x_1}$. Since $f_1^*$ is a minimizer over $\mathcal{H}_{x_1}$ and
        $\hat f\in \mathcal{H}_{x_1}$, we have $Q_k(x_1,f_1^*) \leq Q_k(x_1,\hat f)$. Therefore, it remains to show that
        \begin{equation}
        \label{Q-hat-two-star}
            Q_k(x_1,\hat f)\leq Q_k(x_2,f_2^*).
        \end{equation}
            
        To show inequality \eqref{Q-hat-two-star}, we can write
        \begin{align*}
            Q_k(x_1,\hat{f})
            &= c_{\mathrm{hold}} x_1
            + c_{\mathrm{over}} (\hat{f})^{+}
            + \gamma \mathbb{E}[J_k(x_1 + A - \hat{f} - B)].
        \end{align*}
        Since $x_1\leq x_2$, we have $c_{\mathrm{hold}}x_1\leq c_{\mathrm{hold}}x_2.$
        Also, by construction:
        \[
            \hat{f} = \min\{x_1,2B\}-B < f_2^*,
        \]
        and since $(\cdot)^+$ is non-decreasing, $c_{\mathrm{over}}(\hat{f})^+
            \leq c_{\mathrm{over}}(f_2^*)^+.$
        Finally, we compare the $J_k(\cdot)$ terms that appear on both sides of \eqref{Q-hat-two-star}. We have
        \[
            f_2^*\leq \min\{x_2,2B\}-B\leq x_2-B.
        \]
        Hence, for any realization of $A$, we have $A \leq x_2+A-f_2^*-B$. On the other hand, since $\hat{f}=\min\{x_1,2B\}-B$, we can write the following relations
        \begin{align*}
            &x_1+A-\hat{f}-B=A+x_1-\min\{x_1,2B\}\cr
            &x_2+A-f_2^*-B \geq A+x_2-\min\{x_2,2B\}.
        \end{align*}
        Now note that the function $x-\min\{x,2B\}$ is non-decreasing in $x$. Since $x_1\leq x_2$, we have
        \[
            x_1-\min\{x_1,2B\}
            \leq
            x_2-\min\{x_2,2B\}.
        \]
        Hence, $x_1+A-\hat f-B
            \leq
            x_2+A-f_2^*-B$. Finally since $J_k(\cdot)$ is non-decreasing by the induction hypothesis, we have
        \[
        J_k(x_1+A-\hat f-B) \leq J_k(x_2+A-f_2^*-B). 
        \]
        Since taking expectation preserves the inequality, we can conclude that Eq.~\eqref{Q-hat-two-star} holds and
        \begin{align*}
            J_{k+1}(x_1)-J_{k+1}(x_2) &= Q_k(x_1, f_1^*) - Q_k(x_2, f_2^*) \leq 0.
        \end{align*}
    \end{itemize}

    Hence, in both cases, $J_{k+1}(x_1) \leq J_{k+1}(x_2)$, which completes the induction. Finally, since we have shown that $J_k(\cdot)$ is non-decreasing for all $k$, taking the limit as $k \to \infty$, we conclude that $J^*(\cdot)$ is also a non-decreasing function.

\subsection{Proof of Theorem~\ref{theorem:bangbang}}
\label{proof-bangbang}
 Let $\pi$ denote a stationary deterministic policy. For each state $\mathbf{S} \in \mathcal{S}$ and iteration $t = 1,2, \dots$, let $x^t_{\mathbf{S}}$ be an indicator variable that equals $1$ if the system is in state $\mathbf{S}$ at time step $t$, and $0$ otherwise. Assuming a stationary arrival process where $\mathbf{A}_t = \mathbf{A}$ for all $t=1,2,\dots$, the average eligible arrival rate to the system is given by:
\[
\sum_{\mathbf{S} \in \mathcal{S}} \nu(\mathbf{S})\left(\sum_{i=1}^m q_i \sum_{j\sim \pi(\cdot|\mathbf{S})}^n \mathbf{A}[i,j]\right).
\]
where $\nu(\mathbf{S}) := \lim_{T\to \infty}\frac{1}{T} \sum_{t=1}^T x_{\mathbf{S}}^t$ denotes the state visitation frequency. The block capacity $B$ should therefore satisfy
\[
\sum_{\mathbf{S} \in \mathcal{S}} \nu(\mathbf{S})\left(\sum_{i=1}^m q_i \sum_{j\sim \pi(\cdot|\mathbf{S})}^n \mathbf{A}[i,j]\right) \leq B,
\]
while also making sure that the system remains stable. For the system to remain stable, the average rate of scheduled transactions must equal $Q^t_{\mathrm{arrival}}$, i.e.,
\[
\sum_{s \in \mathcal{S}}\nu(\mathbf{S}) \left(\sum_{i=1}^m q_i \sum_{j\sim \pi(\cdot|\mathbf{S})}^n \mathbf{S}[i,j]\right) = Q^t_{\mathrm{arrival}}.
\]
Using the definition of the occupancy measure $\rho(\mathbf{S}, a)$ as the average proportion of time the system encounters the state-action pair $(\mathbf{S}, a)$ under  policy $\pi$ \cite{altman2021constrained,qin2023scalable}, i.e., 
\begin{align}
    \rho(\mathbf{S}, a) &=\lim_{T\to \infty}\frac{1}{T}\sum_{t=0}^{T} \mathbb{P}(\mathbf{S}^t = \mathbf{S}, a^t = a)\nonumber \\& = \nu(\mathbf{S})\pi(a|\mathbf{S})\label{eq:occupancy-measure-define},
\end{align}
we can formulate an optimization problem to determine the minimum block capacity required to accommodate all incoming transactions as below: 
\begin{align}
\label{initialObjective}
\min \quad 
& \sum_{a \in \mathcal{A}}\sum_{\mathbf{S} \in \mathcal{S}} \rho(\mathbf{S}, a)\left(\sum_{i=1}^m q_i \sum_{j=a}^n \mathbf{A}[i,j]\right)  \\
\mathrm{s.t.} \quad 
& \rho(\mathbf{S},a) \in \Delta,\\
& \sum_{a \in \mathcal{A}}\sum_{\mathbf{S} \in \mathcal{S}} \rho(\mathbf{S}, a)\left(\sum_{i=1}^m q_i \sum_{j=a}^n \mathbf{S}[i,j]\right) = Q^t_{\mathrm{arrival}}.
\end{align}

To further simplify the analysis, we assume that exactly one transaction of each type $(q_i, v_j)$ arrives at every time step. Under this assumption, the above optimization problem reduces to:
\begin{align*}
\min \quad 
& \sum_{a \in \mathcal{A}}\sum_{\mathbf{S} \in \mathcal{S}} \rho(\mathbf{S}, a) \cdot \left(n - a + 1\right) \cdot \sum_{i=1}^m q_i  \\
\mathrm{s.t.} \quad 
& \rho(\mathbf{S},a) \in \Delta,\\
& \sum_{a \in \mathcal{A}}\sum_{\mathbf{S} \in \mathcal{S}} \rho(\mathbf{S}, a)\left(\sum_{i=1}^m q_i \sum_{j=a}^n \mathbf{S}[i,j]\right) = Q^t_{\mathrm{arrival}}.
\end{align*}
This objective function can be further simplified to:
\begin{align*}
\max \quad 
& \sum_{a \in \mathcal{A}}\sum_{\mathbf{S} \in \mathcal{S}} a \cdot \rho(\mathbf{S}, a) \\
\mathrm{s.t.} \quad 
& \rho(\mathbf{S},a) \in \Delta,\\
& \sum_{a \in \mathcal{A}}\sum_{\mathbf{S} \in \mathcal{S}} \rho(\mathbf{S}, a)\left(\sum_{i=1}^m q_i \sum_{j=a}^n \mathbf{S}[i,j]\right) = Q^t_{\mathrm{arrival}}.
\end{align*}

Next, to obtain a lower bound for $B$, it is enough to construct a feasible solution following bang--bang mechanism to this optimization problem, whose objective value will serve as a lower bound for $B$. We do so by constructing a stationary deterministic policy $\pi$ and then defining a corresponding occupancy measure $\rho\in\Delta$ induced by $\pi$. Fix an integer $k\ge 1$ and define the following state
\[
\mathbf{S}^{(k)} :=
\begin{pmatrix}
k & k & \cdots & 1\\
k & k & \cdots & 1\\
\vdots & \vdots & \ddots & \vdots\\
k & k & \cdots & 1
\end{pmatrix}.
\]
We consider the stationary deterministic policy $\pi$ following the bang--bang mechanism given by
\begin{equation}
\label{eq:policy-pi-k}
    \pi(\mathbf S) :=
    \begin{cases}
    1, & \text{if } \mathbf S = \mathbf S^{(k)},\\
    n, & \text{otherwise.}
    \end{cases}
\end{equation}

That is, the policy chooses the highest price $a=n$ in all states except when the pool reaches $\mathbf S^{(k)}$, at which point it posts the lowest price $a=1$ to make the pool less congested.

Under our deterministic arrival assumption, each step adds a total size 
$Q := \sum_{i=1}^m q_i$ to each value column.
When $\pi(\mathbf S) = n$, only the top-value column is eligible; consequently, the first $n-1$ columns increase by $+1$ in every entry at each time step.
Starting from the state $\mathbf{S}^{(1)}$, which is obtained immediately after the first arrival, after $k-1$ consecutive steps of playing $a = n$, the state becomes exactly $\mathbf S^{(k)}$.
At that point, the policy plays $a = 1$, making all columns eligible.

The total pool size at $\mathbf S^{(k)}$ equals
\[
\sum_{i=1}^m q_i\bigl((n-1)k+1\bigr) \;=\; \bigl((n-1)k+1\bigr)Q.
\]
Thus, if
\begin{equation}
\label{eq:flush-B}
\bigl((n-1)k+1\bigr)Q \le 2B,
\end{equation}
then, during the flush step, the scheduler is able to clear the entire mempool. Immediately afterward, one transaction of each type arrives, returning the system to the state $\mathbf S^{(1)}$.

Now fix the above stationary policy $\pi$ and consider the induced MDP with transition kernel
$P(\mathbf S'\mid \mathbf S,\pi(\mathbf S))$. Note that by choosing the target block capacity sufficiently large, as in Eq.~\eqref{eq:flush-B}, the state evolution becomes periodic over a finite set of states, with each cycle closing by returning to the $\mathbf S^{(1)}$ state. This periodic behavior is illustrated in Fig.~\ref{fig:state-dynamics}.

\begin{figure}[htb]
    \centering
    \includegraphics[width=0.9\linewidth]{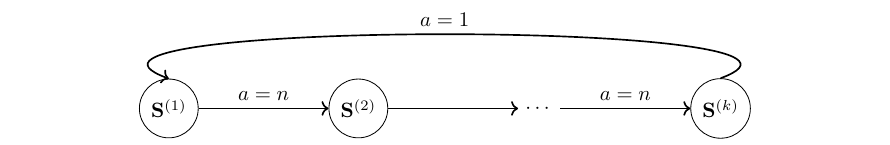}
    \caption{Periodic state evolution induced by the policy defined in Eq.~\eqref{eq:policy-pi-k}: $k-1$ steps with $a=n$ followed by one flush step with $a=1$.}
    \Description{Periodic state evolution induced by the policy defined in Eq.~\eqref{eq:policy-pi-k}: $k-1$ steps with $a=n$ followed by one flush step with $a=1$.}
    \label{fig:state-dynamics}
\end{figure}
A direct result of the periodic state evolution described above is that  the induced Markov chain visits each state in the cycle exactly once per period. Hence for $\mathbf S\in\{\mathbf S^{(1)},\ldots,\mathbf S^{(k)}\}$, the induced state distribution can be written as  
\[
\nu(\mathbf S)=\frac{1}{k},
\]
and $\nu(\mathbf S)=0$ for all other $\mathbf S\in\mathcal{S}$ not belonging to this cycle. Therefore, the occupancy measure defined in \eqref{eq:occupancy-measure-define} corresponding to this policy is given by
\[
\rho\bigl(\mathbf S^{(\ell)}, n\bigr) = \frac{1}{k}
\quad \text{for } \ell = 1,\ldots,k-1,
\qquad
\rho\bigl(\mathbf S^{(k)}, 1\bigr) = \frac{1}{k},
\]
and $\rho(\mathbf S, a) = 0$ otherwise. We can then verify that $\rho(\mathbf S, a) \in \Delta$.
Substituting the policy defined in Eq.~\eqref{eq:policy-pi-k} into the initial objective in Eq.~\eqref{initialObjective}, we get
\[
\frac{n+k-1}{k} \cdot Q \leq B,
\]
which then combined with inequality Eq.~\eqref{eq:flush-B} gives us
\begin{align}
\label{k-range}
    \frac{Q(n-1)}{B - Q} \leq k \leq \frac{2B-Q}{Q(n-1)},
\end{align}
where $Q=\sum_{i=1}^m q_i$.
Note that the above statement should hold while making sure that $k$ remains an integer. Therefore we need to find the smallest $B$ such that 
\[
\frac{Q(n-1)}{B - Q} \leq \frac{2B-Q}{Q(n-1)},
\]
which further simplifies to
\[
0 \leq 2B^2 -3BQ+2nQ^2-Q^2n^2,
\]
and finally
\begin{align}
\label{eq:B-lower-bound}
\frac{nQ}{2}<
    \frac{Q}{4}\left(3+\sqrt{8n^2-16n+9}\right) \leq B.
\end{align}
 
Thus, for all $B$ satisfying the above inequality, it is sufficient to find the smallest value of $B$ for which the range shown in Eq.~\eqref{k-range} includes a positive integer $k$.

Now that we have established a lower bound on the target block capacity, we want to derive an upper bound on the \emph{average overshoot cost} defined in Eq.~\eqref{eq:overshoot-cost}. Over a cycle of $k$ steps, exactly $k-1$ blocks are scheduled without incurring any overshoot cost. Consequently, the average overshoot cost over the cycle satisfies
\[
\overline{C}_{\mathrm{over}} = \frac{((n-1)k+1)Q - B}{k}.
\]
Using the result of Remark~\ref{B-Q-range}, we know $\frac{nQ}{2}\leq B< nQ$, and since $n\geq2$, we can conclude that $Q\leq B$. Therefore, we can conclude that $\overline{C}_{\mathrm{over}}$ is a non-decreasing function of $k$. By substituting the upper bound on $k$ in Eq.~\eqref{k-range} we get
\[
\overline{C}_{\mathrm{over}}
\leq \frac{BQ(n-1)}{2B -Q}.
\]

\section{Additional Numerical Results}\label{appendix:numerical}    
\subsection{Numerical Results for Algorithm~\ref{alg:NPG}}

In this subsection, we present numerical results for Settings~2 and~3 under stochastic arrivals. Figure~\ref{fig:Stochastic-Setting2} illustrates the average scheduled transaction volume as a function of the overshoot penalty for three arrival rates $\lambda$ and three target block capacities in Setting~2. As in the deterministic arrival setting, the average scheduled transaction volume approaches the target block capacity as the overshoot penalty increases. However, for larger arrival rates, transaction counts exceed the threshold $L$ more frequently, requiring the algorithm to prioritize scheduling the excess transactions. Consequently, the algorithm has fewer opportunities to regulate the scheduled transaction volume, making it difficult to keep the average scheduled transaction volume close to the target block capacity. 

Figure~\ref{fig:Stochastic-Setting3} presents the results for Setting~3. Although Settings~2 and~3 are fundamentally different, with Setting~2 having larger transaction sizes and Setting~3 having higher transaction values, the proposed algorithm performs well in both settings. This demonstrates that Algorithm~\ref{alg:NPG} is robust to different transaction characteristics.

\begin{figure}[htb]
    \centering
    \begin{subfigure}{0.32\textwidth}
        \centering
        \includegraphics[width=\linewidth]{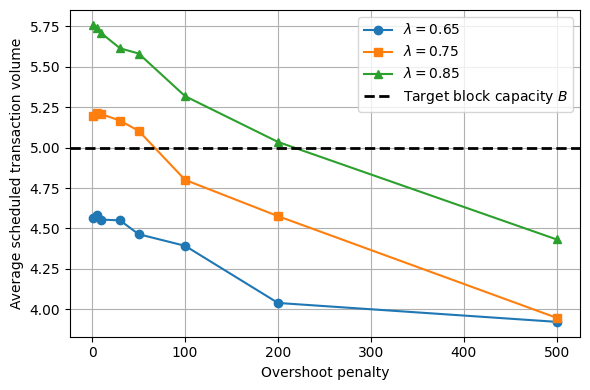}
        \caption{$B=5$}
        \label{fig:S2,B=5}
    \end{subfigure}
    \hfill
    \begin{subfigure}{0.32\textwidth}
        \centering
        \includegraphics[width=\linewidth]{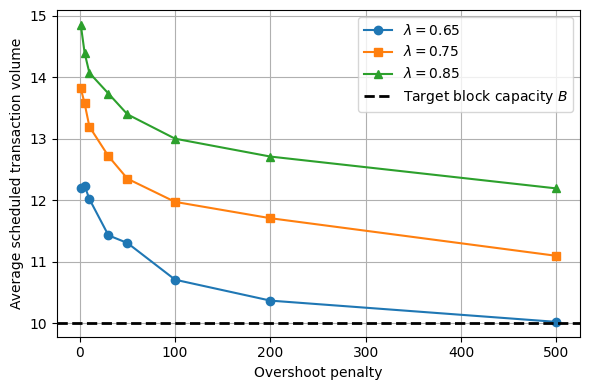}
        \caption{$B=10$}
        \label{fig:S2,B=10}
    \end{subfigure}
    \hfill
    \begin{subfigure}{0.32\textwidth}
        \centering
        \includegraphics[width=\linewidth]{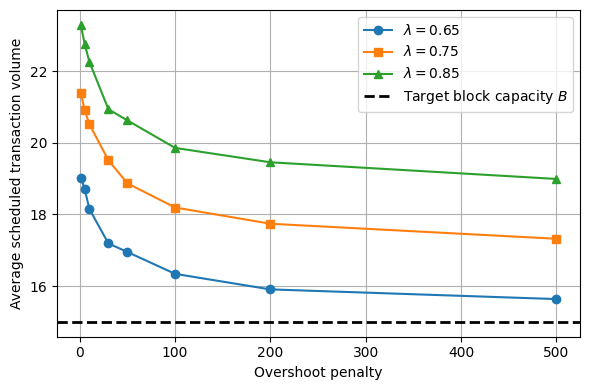}
        \caption{$B=15$}
        \label{fig:S2,B=15}
    \end{subfigure}
    \caption{Average scheduled transaction volume as a function of the overshoot penalty for Setting~2 under stochastic arrivals with three different target block capacities.}
    \Description{Average scheduled transaction volume as a function of the overshoot penalty for Setting~2 under stochastic arrivals with three different target block capacities.}
    \label{fig:Stochastic-Setting2}
\end{figure}

\begin{figure}[htb]
    \centering
    \begin{subfigure}{0.32\textwidth}
        \centering
        \includegraphics[width=\linewidth]{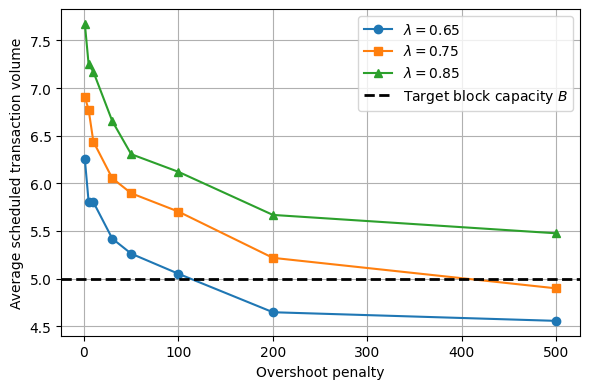}
        \caption{$B=5$}
        \label{fig:S3,B=5}
    \end{subfigure}
    \hfill
    \begin{subfigure}{0.32\textwidth}
        \centering
        \includegraphics[width=\linewidth]{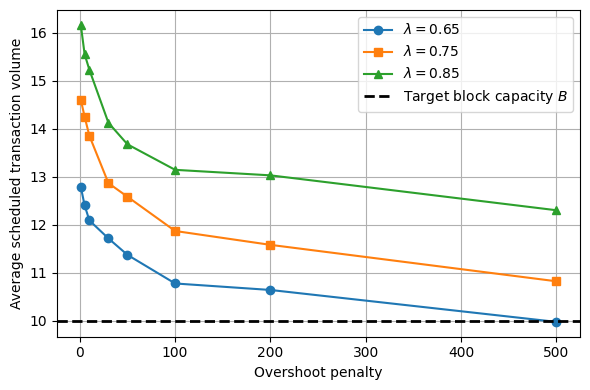}
        \caption{$B=10$}
        \label{fig:S3,B=10}
    \end{subfigure}
    \hfill
    \begin{subfigure}{0.32\textwidth}
        \centering
        \includegraphics[width=\linewidth]{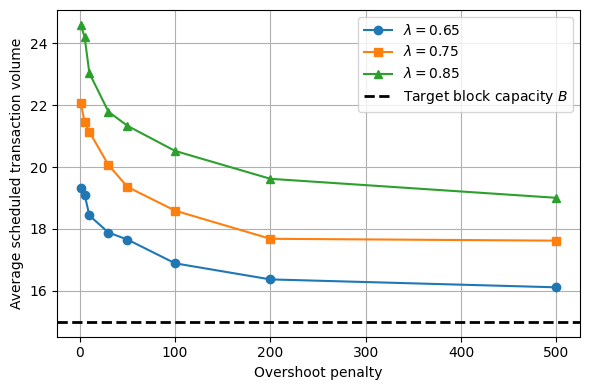}
        \caption{$B=15$}
        \label{fig:S3,B=15}
    \end{subfigure}
    \caption{Average scheduled transaction volume as a function of the overshoot penalty for Setting~3 under stochastic arrivals with three different target block capacities.}
    \Description{Average scheduled transaction volume as a function of the overshoot penalty for Setting~3 under stochastic arrivals with three different target block capacities.}
    \label{fig:Stochastic-Setting3}
\end{figure}

\subsection{Evaluation of the Bang--Bang Pricing Mechanism}
Here, we numerically validate the average overshoot cost upper bound established in Theorem~\ref{theorem:bangbang}. Let $B_{\mathrm{th}}$ denote the theoretical lower bound on the target block capacity derived in Eq.~\eqref{eq:B-lower-bound-theorem}. We consider target block capacities $B$ ranging from $B_{\mathrm{th}}$ to $2B_{\mathrm{th}}$.

For each value of $B$, we first determine the smallest positive integer $k$ satisfying the feasibility condition in Eq.~\eqref{k-range-theorem}, and then simulate the corresponding bang--bang pricing policy described in Eq.~\eqref{eq:policy-pi-k} under uniform arrivals. We evaluate the long-run average overshoot cost defined in Eq.~\eqref{eq:overshoot-cost}. Figure~\ref{fig:bangbang} illustrates the average overshoot cost as a function of $B/B_{\mathrm{th}}$ for three different parameter settings, and compares the empirical results with the theoretical upper bound derived in Theorem~\ref{theorem:bangbang}. As expected, increasing the target block capacity reduces the average overshoot cost, which eventually becomes zero once the block capacity is sufficiently large to accommodate the arriving workload without exceeding the target. The three parameter settings correspond to different choices of the transaction sizes and the number of transaction value levels. Specifically, let $\mathcal{Q}$ denote the set of transaction sizes, and let $n$ denote the number of distinct transaction values.

\begin{figure}[htb]
    \centering
    \begin{subfigure}{0.32\textwidth}
        \centering
        \includegraphics[width=\linewidth]{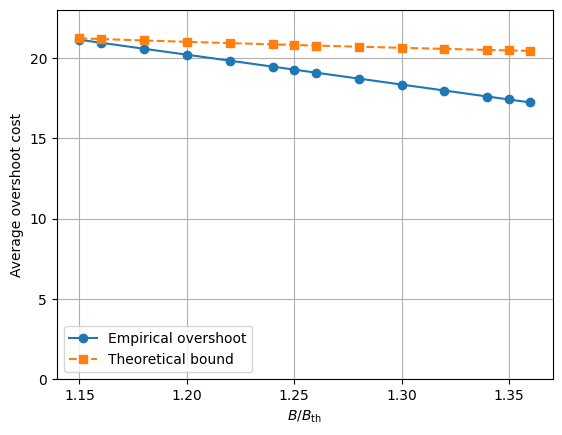}
        \caption{$\mathcal{Q}=\{2,4,5,7\},\ n=3$.}
        \label{fig:bangbangS1}
    \end{subfigure}
    \hfill
    \begin{subfigure}{0.32\textwidth}
        \centering
        \includegraphics[width=\linewidth]{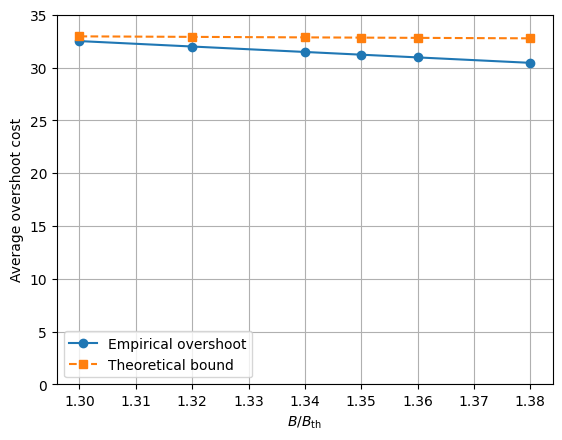}
        \caption{$\mathcal{Q}=\{1,1,2,2,3,3\},\ n=6$.}
        \label{fig:bangbangS2}
    \end{subfigure}
    \hfill
    \begin{subfigure}{0.32\textwidth}
        \centering
        \includegraphics[width=\linewidth]{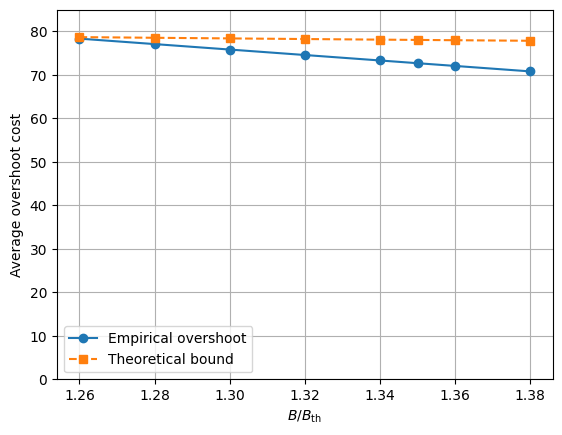}
        \caption{$\mathcal{Q}=\{5,8,10,12\},\ n=5$.}
        \label{fig:bangbangS3}
    \end{subfigure}
    \caption{Average overshoot cost as a function of $B/B_{\mathrm{th}}$ under the bang--bang pricing policy for three different parameter settings. Throughout these experiments, we set $c_{\mathrm{over}}=1$.}
    \Description{}
    \label{fig:bangbang}
\end{figure}
\end{document}